\RequirePackage[l2tabu,orthodox]{nag}
\documentclass[format=acmsmall, nonacm=true, review=false, screen=true]{acmart}

\usepackage[utf8]{inputenc}

\usepackage{amsmath}
\usepackage{amsthm}
\usepackage{bm}
\usepackage{mathtools}
\usepackage{xspace}

\usepackage[inline]{enumitem}
\setlist[description]{labelindent=\parindent}
\newenvironment{algorr}[2]{%
	\bigskip
	\noindent\textbf{Algorithm #1:} {\itshape#2}
	\begin{enumerate}[leftmargin=1.5em,label=\texttt{\arabic*}]%
		\smallskip
		\setlength{\itemsep}{0pt}\leftmargin=0pt}{%
	\end{enumerate}%
}

\newtheorem{thm}{Theorem}
\newtheorem{defn}[thm]{Definition}
\newtheorem{lem}[thm]{Lemma}

\newtheorem*{rep@theorem}{\rep@title}
\newcommand{\newreptheorem}[2]{%
	\newenvironment{rep#1}[1]{%
    \def\rep@title{#2~\ref{##1} (restated)}%
		\begin{rep@theorem}}%
		{\end{rep@theorem}}}
\newreptheorem{thm}{Theorem}
\newreptheorem{lem}{Lemma}
\newreptheorem{prop}{Proposition}
\newreptheorem{cor}{Corollary}

\DeclarePairedDelimiter\paren{\lparen}{\rparen}
\DeclarePairedDelimiter\abs{\lvert}{\rvert}
\DeclarePairedDelimiter\set{\{}{\}}
\DeclarePairedDelimiterX\setc[2]{\{}{\}}{#1:#2}
\DeclarePairedDelimiterX\parenc[2]{\lparen}{\rparen}{\,#1 \;\delimsize\vert\; #2\,}

\usepackage{empheq}
\usepackage{tcolorbox}
\usepackage{varwidth}
\definecolor{blueish}{rgb}{0.122, 0.435, 0.698}
\definecolor{dagstuhlyellow}{rgb}{0.99,0.78,0.07}
  
\newcommand{\defproblem}[3]{
\vspace{4pt}\noindent
\begin{tcolorbox}[%
  hbox,
  colframe=white,
  colback=dagstuhlyellow!10!white,
  arc=10pt]
  \begin{varwidth}{0.9\textwidth}
    {\textbf{Problem}} #1 expects as input: #2\\
    {\textbf{Task:}} #3
  \end{varwidth}
\end{tcolorbox}
}

\newcommand{\counter}{\textnormal{\texttt{EdgeCount}}}
\newcommand{\cnfsat}[1]{\textnormal{\ensuremath{#1}-SAT}\xspace}
\newcommand{\ccnfsat}[1]{\textnormal{\#\ensuremath{#1}-SAT}\xspace}

\newcommand{\SAT}[1]{\textnormal{SAT}{(#1)}}
\newcommand{\pr}{\mathbb{P}}
\newcommand{\E}{\mathbb{E}}
\newcommand{\var}{\textnormal{Var}}

\newcommand{\cc}[1]{\ensuremath{\mathrm{#1}}}
\newcommand{\pp}[1]{\textup{#1}}
\newcommand{\op}[1]{\ensuremath{\operatorname{#1}}}

\newcommand{\NP}{\cc{NP}}

\newcommand{\sharpP}{\cc{\#P}}
\newcommand{\poly}{\op{poly}}
\newcommand{\GF}[1]{\op{GF}{(#1)}}

\newcommand{\N}{\mathbb{N}}
\newcommand{\Z}{\mathbb{Z}}

\newcommand{\R}{\mathbb{R}}

\renewcommand{\epsilon}{\varepsilon}
\renewcommand{\vec}[1]{\boldsymbol{#1}} 
\newcommand{\floor}[1]{\left\lfloor{}#1\right\rfloor}
\newcommand{\ceil}[1]{\left\lceil{}#1\right\rceil}
\newcommand{\dov}{\pp{OV}}
\newcommand{\cov}{\pp{\#OV}}
\newcommand{\dnwt}{\pp{NWT}}
\newcommand{\cnwt}{\pp{\#NWT}}
\newcommand{\dtsum}{\pp{3SUM}}
\newcommand{\ctsum}{\pp{\#3SUM}}
 
\newcommand{\aref}[1]{Step~\ref{#1}} 
\newcommand{\indo}[1]{\textnormal{ind}_{#1}} 
\newcommand{\adjo}[1]{\textnormal{adj}_{#1}} 
\newcommand{\eb}[1]{\partial(#1)}            
                         
\newcommand{\pf}[1]{{#1}^*}
\newcommand{\nbh}[1]{N{(#1)}}

\begin{document}

\title[Fine-grained reductions from approximate counting to decision]{\texorpdfstring{Fine-grained reductions\\from approximate counting to decision}{Fine-grained reductions from approximate counting to decision}}

\author{Holger Dell}
\orcid{0000-0001-8955-0786}
\affiliation{%
	\institution{Goethe University Frankfurt}
	\city{Frankfurt}
	\country{Germany}
  }
\affiliation{%
	\institution{IT University of Copenhagen}
	\city{Copenhagen}
	\country{Denmark}
  }
\affiliation{%
  \institution{Basic Algorithms Research Copenhagen (BARC)}
  \city{Copenhagen}
  \country{Denmark}
}
\email{dell@cs.uni-frankfurt.de}

\author{John Lapinskas}
\affiliation{%
  \institution{University of Bristol}
  \city{Bristol}
  \country{UK}
}
\email{john.lapinskas@bristol.ac.uk}



\begin{abstract}
	In this paper, we introduce a general framework for fine-grained reductions of approximate counting problems to their decision versions. (Thus we use an oracle that decides whether any witness exists to multiplicatively approximate the number of witnesses with minimal overhead.) This mirrors a foundational result of Sipser (STOC 1983) and Stockmeyer~(SICOMP 1985) in the polynomial-time setting, and a similar result of Müller~(IWPEC 2006) in the FPT setting. Using our framework, we obtain such reductions for some of the most important problems in fine-grained complexity: the Orthogonal Vectors problem, 3SUM, and the Negative-Weight Triangle problem (which is closely related to All-Pairs Shortest Path). While all these problems have simple algorithms over which it is conjectured that no polynomial improvement is possible, our reductions would remain interesting even if these conjectures were proved; they have only polylogarithmic overhead, and can therefore be applied to subpolynomial improvements such as the $n^3/\exp\paren{\Theta(\sqrt{\log n})}$-time algorithm for the Negative-Weight Triangle problem due to Williams~(STOC~2014). Our framework is also general enough to apply to versions of the problems for which more efficient algorithms are known. For example, the Orthogonal Vectors problem over $\GF{m}^d$ for constant~$m$ can be solved in time $n\cdot\poly(d)$ by a result of Williams and Yu (SODA 2014); our result implies that we can approximately count the number of orthogonal pairs with essentially the same running time.
		
  We also provide a fine-grained reduction from approximate \#SAT to SAT. Suppose the Strong Exponential Time Hypothesis (SETH) is false, so that for some $1<c<2$ and all $k$ there is an $O(c^n)$-time algorithm for \cnfsat{k}. Then we prove that for all $k$, there is an $O(\paren{c+o(1)}^n)$-time algorithm for approximate \#$k$-SAT. In particular, our result implies that the Exponential Time Hypothesis (ETH) is equivalent to the seemingly-weaker statement that there is no algorithm to approximate \#$3$-SAT to within a factor of $1+\epsilon$ in time $2^{o(n)}/\epsilon^2$ (taking $\epsilon > 0$ as part of the input).
\end{abstract}

\maketitle

	\section{Introduction}\label{sec:intro}
	
	It is clearly at least as hard to count objects as it is to decide their existence, and often it is harder.
  For a concrete example, there is a polynomial-time algorithm to find a perfect matching in a bipartite graph if one exists, but computing the exact number of all perfect matchings is a $\sharpP$-complete problem~\cite{ValPerm}, which means that solving it in polynomial time would collapse the polynomial-time hierarchy~\cite{Toda}.
  	However, the situation changes substantially if we consider approximate rather than exact counting. For all real~$\epsilon$ with $0 < \epsilon < 1$, we say that $x \in \R$ is an \emph{$\epsilon$-approximation} to $N\in\R$ if $|x-N| \le \epsilon N$ holds.
  Since the approximation guarantee is multiplicative, computing an $\epsilon$-approximation to $N$ is at least as hard as deciding whether $N > 0$ holds. In fact, these two tasks are often roughly equally hard, and indeed this is true for our example:
  Jerrum, Sinclair, and Vigoda~\cite{DBLP:journals/jacm/JerrumSV04} proved that an $\epsilon$-approximation to the number of perfect matchings in a bipartite graph can be computed in polynomial time.
  While there is a polynomial-time algorithm to find perfect matchings in bipartite graphs and one to approximately count them, there is still an important discrepancy: The former algorithm runs in quasi-linear time while the latter runs in time~$\epsilon^{-2}\cdot\tilde{O}(n^{10})$.
  
  This paper is concerned with fine-grained complexity, in which one considers the exact running time of an algorithm rather than broad categories such as polynomial time, FPT time, or subexponential time.   
  Reductions that solve an approximate counting problem by means of an oracle for its decision version have been studied already in various different contexts.
  Sipser~\cite{Sipser} and Stockmeyer~\cite{Stockmeyer} proved implicitly that every problem in $\sharpP$ has a polynomial-time randomised $\epsilon$-approximation algorithm that has access to an $\NP$-oracle; the result is later made explicit by Valiant and Vazirani~\cite{VV}.
  In parameterised complexity, M\"{u}ller~\cite{MullerCounting} proved an analogue of this result for the W-hierarchy: In particular, for every problem in \#W[1], there is a randomised algorithm that has access to some W[1]-oracle, runs in time $f(k)\cdot\poly(n,\epsilon^{-1})$ for some computable $f:\N\rightarrow\N$, and outputs an $\epsilon$-approximation to the problem.
  Finally, in the exponential-time setting, Thurley~\cite{DBLP:conf/stacs/Thurley12} proposed a reduction for $k$-SAT that implies: If there is an $\pf{O}(2^{(1-\delta)n})$-time algorithm for $k$-SAT for some~${\delta > 0}$, then there is an $\epsilon$-approximation algorithm for \#$k$-SAT that runs in time
  $\epsilon^{-2}\cdot\pf{O}(2^{(1-\delta/2)n})$. (This reduction was later improved by Schmitt and Wanka~\cite{wanka-ksat}.) Such results are an important foundation of the wider complexity theory of approximate counting initiated by Dyer, Goldberg, Greenhill and Jerrum~\cite{dggj-approx}. However, all of these reductions introduce significant overheads to the running time --- they are not fine-grained. 
  
  Perhaps the most important polynomial-time problems in fine-grained complexity are orthogonal vectors (OV), 3SUM, and all-pairs shortest paths (APSP). All three problems admit well-studied notions of hardness, in the sense that many problems reduce to them or are equivalent to them under fine-grained reductions, and they are not known to reduce to one another. See Vassilevska~Williams~\cite{Williams-3sum} for a recent survey. It is not clear what a ``canonical'' counting version of APSP should be, but it is equivalent to the Negative-Weight Triangle problem (NWT) under subcubic reductions~\cite{WW-NWT}, so we consider this instead. We give highly efficient fine-grained reductions from approximate counting to decision for all three problems. All of these results are immediate corollaries of an algorithm which counts edges in a bipartite graph to which it has limited oracle access; this algorithm has several additional applications, including some new approximate counting algorithms for related problems. We discuss our edge-counting framework further in Section~\ref{sec:framework}, and describe its applications in Section~\ref{sec:introP} together with a detailed overview of the~literature.
  
  The most important exponential-time problem in fine-grained complexity is unequivocally SAT. We provide a fine-grained reduction from approximate \#$k$-SAT to $O(k\log^2 k)$-SAT as $k\to\infty$; as a corollary, we show that if the Strong Exponential Time Hypothesis (SETH) is false, then the savings from decision $k$-SAT as $k\to\infty$ can be passed on to approximate \#$k$-SAT with subexponential overhead. Our reduction also implies that the Exponential Time Hypothesis (ETH) is equivalent to an approximate counting version. We discuss the reduction and its corollaries further in Section~\ref{sec:sat-intro}.

  \subsection{Approximately Counting Edges in Bipartite Graphs}\label{sec:framework}
  
  Let $G$ be a bipartite graph with~$G=(U,V,E)$.
	We consider a computation model where the algorithm is given~$U$ and~$V$, and can access the edges of the graph only via its adjacency oracle and its independence oracle:
  \begin{itemize}
    \item
      The \emph{adjacency oracle} of $G$ is the function $\adjo{G}:U\times V\rightarrow\{0,1\}$ such that $\adjo{G}(u,v) = 1$ if and only if $(u,v) \in E$.
    \item
      The \emph{independence oracle} of $G$ is the function $\indo{G}:2^{U\cup V} \rightarrow \{0,1\}$ such that $\indo{G}(S) = 1$ if and only if $S$ is an independent set in $G$.
  \end{itemize}
  Of course, the adjacency oracle can be simulated with the independence oracle by querying sets of two vertices.
  We distinguish them here, because we wish to think of independence queries as very expensive, and we will use them only polylogarithmically often.
  Our main result is as follows:
	\def\statefinegrain{
	There is a randomised algorithm $\mathcal{A}$ which, given a rational number $\epsilon$ with $0 < \epsilon < 1$ and oracle access to an $n$-vertex bipartite graph $G$, outputs an $\epsilon$-approximation of $|E(G)|$ with probability at least $2/3$. Moreover,~$\mathcal{A}$ runs in time $\epsilon^{-2}\cdot O(n\log^4 n\log\log n)$ and makes at most ${\epsilon^{-2}\cdot O(\log^5 n\log\log n)}$ calls to the independence oracle.
	}
	\begin{thm}\label{thm:finegrain}
		\statefinegrain{}
	\end{thm}
	We prove this result in Section~\ref{sec:finegrain}.
  Note that since oracle calls are constant time operations, the adjacency oracle is called at most $\epsilon^{-2}\cdot\tilde{O}(n)$ times. Moreover, a polynomial factor of~$\epsilon^{-1}$ in the running time is to be expected, since the exact value of $|E(G)|$ can be recovered by taking~$\epsilon = 1/(2n^2)$.
  
  In independent work, Beame et al.~\cite{new-oracle} obtain a result similar to Theorem~\ref{thm:finegrain}, with an overall running time of $\epsilon^{-4}\cdot\tilde{O}(1)$ but with no further bound on the number of independence queries used. Thus their result outperforms Theorem~\ref{thm:finegrain} when independence queries are fast, and underperforms when they are slow. In all our applications, independence queries are so slow as to dominate our running times; thus substituting Beame et al.'s result for Theorem~\ref{thm:finegrain} would yield worse algorithms.
  
  While Theorem~\ref{thm:finegrain} is not able to deal with the non-bipartite case at all, Beame et al.~\cite{new-oracle} present a second algorithm in their paper that is able to approximately count edges in general graphs by using $\tilde{O}(n^{2/3})$ queries to the independence oracle.
  Recently, Chen, Levi, and Waingarten~\cite{DBLP:conf/soda/0001LW20} improve the number of queries to $\tilde{O}(\sqrt{n})$ and prove unconditionally that this is optimal for general graphs.
  
  In work subsequent to this paper, using a different technique, the authors and Meeks~\cite{DBLP:conf/soda/DellLM20} were able to remove the adjacency queries from Theorem~\ref{thm:finegrain} while retaining a bound of $\epsilon^{-2} \log^{O(1)} n$ for the number of independence queries. Moreover, they generalise the theorem to $k$-partite $k$-uniform hypergraphs, where the bound on the number of queries is at most $\epsilon^{-2} \log^{O(k)} n$, and extend the result to cover approximately-uniform sampling. This generalisation has consequences for the fine-grained complexity of problems that do not directly correspond to bipartite graphs, such as approximately counting graph motifs.
  Independently, Bhattacharya et al.~\cite{DBLP:journals/corr/abs-1808-00691,DBLP:journals/corr/abs-1908-04196} generalise Theorem~\ref{thm:finegrain} to $k$-partite $k$-uniform hypergraphs, obtaining a somewhat weaker bound of $\epsilon^{-4} \log^{O(k)} n$ on the number of queries.
  Moreover, Bishnu et al.~\cite{DBLP:conf/isaac/BishnuGKM018} use the generalised oracle to solve various decision problems in parameterised complexity.

  \subsection{Corollaries for Problems in P}\label{sec:introP}
  As described in Section~\ref{sec:intro}, the problems Orthogonal Vectors (\dov), \dtsum, and Negative-Weight Triangle (NWT) are central players in the field of fine-grained complexity. All three problems have simple polynomial-time exhaustive-search algorithms over which it is conjectured that no truly polynomial improvement is possible.
  The same exhaustive search algorithms also solve the canonical counting versions of these problems.
  Nevertheless it is possible that the decision version has faster algorithms while the exact counting version does not. 
  Our results imply that any improvement to decision algorithms transfers to the approximate counting version of the problem as well, up to polylogarithmic factors in the running time.
  
  In fact, for OV~\cite{AWY15} and~NWT~\cite{Williams-APSP}, non-trivial (subpolynomial) improvements over exhaustive search algorithms are already known.
  Our results transfer these improvements to approximate counting.
  In the case of the standard version of OV, this turns out to be uninteresting as the derandomisation of~\cite{AWY15} due to Chan and Williams~\cite{CW-Counting} already solves the \emph{exact} counting version.
  However in the case of NWT, we are not aware of improved algorithms for the counting version; using our reduction, we obtain such an algorithm for approximate counting. Our reductions also apply to several variants of the three central problems, yielding more new algorithms. Notably, for one variant of OV we obtain a quasilinear-time approximate counting algorithm, but all exact counting algorithms require quadratic time under SETH.
  In the following, we state our results~formally.

  \subsubsection{OV}

  In the orthogonal vectors problem OV, we are given two lists $A$ and $B$ of zero-one vectors over~$\R^d$, and must determine whether there exists an orthogonal pair $(a,b) \in A\times B$. In \#OV, we must instead determine the number of orthogonal pairs. Writing $n=|A|+|B|$, it is easy to see that \dov\ and \#OV can both be solved in $O(n^2d)$ operations by iterating over all pairs.
  The \emph{low-dimension OV conjecture}~\cite{Williams-OVSETH,DBLP:conf/soda/GaoIKW17} asserts that in the case where ${d=\omega(\log n)}$, there is no randomised algorithm that solves \dov\ in time $O(n^{2-\delta})$, for any constant $\delta>0$.
  This conjecture is implied by the Strong Exponential Time Hypothesis (SETH)~\cite{Williams-OVSETH}, and Abboud, Williams, and Yu~\cite{AWY15} proved that it fails when $d=O(\log n)$.

  To reduce the approximate version of \#OV to OV, we model the instance as a bipartite graph and apply the edge estimation algorithm from Theorem~\ref{thm:finegrain}.
  Indeed, the list $A$ becomes the left side of the graph, $B$ the right side, and each orthogonal pair $(a,b)$ becomes an edge.
  Then approximately counting orthogonal pairs reduces to estimating the number of edges in this graph, adjacency queries take time $O(d)$ and correspond to computing the inner product of two vectors, and independence queries are simulated by invoking the assumed decision algorithm for OV.
  In this way, in Section~\ref{sec:ov-proofs} we obtain the following structural complexity result as a corollary to Theorem~\ref{thm:finegrain}.
  \newcommand{\stateov}{
    If $\dov$ with~$n$ vectors in $d$ dimensions has a randomised algorithm that runs in time $T(n,d)$, then there is a randomised $\epsilon$-approximation algorithm for~$\cov$ that runs in time $T(n,d)\cdot \epsilon^{-2}O(\log^6 n\log\log n)$.
  }
  \begin{thm}\label{thm:ov}
    \stateov{}
  \end{thm}
  
  In particular, if $\epsilon^{-1}$ is at most polylogarithmic in~$n$, Theorem~\ref{thm:ov} implies we can \mbox{$\epsilon$-approximate} \#OV with only polylogarithmic overhead over decision.
  
  While OV has a non-trivial algorithm~\cite{AWY15} with running time $n^{2-1/O(\log(d/\log n))}$ as we mentioned in Section~\ref{sec:introP}, it has already been adapted into an exact \#OV algorithm with the same running time~\cite{CW-Counting}, so Theorem~\ref{thm:ov} does not yield a new algorithm at the moment. However, any further improvement for the decision version of the problem will immediately translate to a new approximate counting algorithm.

  Interestingly, there is a variant of \dov\ for which our method does yield a new algorithm; in this variant, the real zero-one vectors are replaced by arbitrary vectors over a finite field or over the integers modulo $m$.
  Even though Williams and Yu~\cite{WY-OV-algo} did not consider the counting version and their algorithms do not seem to generalise to counting, we can nevertheless use their decision algorithm as a black box to obtain an efficient approximate counting algorithm as a corollary to Theorem~\ref{thm:finegrain}. 
	
	\newcommand{\stateovalgo}{
		Let $m=p^k$ be a constant prime power. There is a randomised $\epsilon$-approxi\-mation algorithm for \cov\ over $\GF{m}^d$ with running time $\epsilon^{-2}d^{(p-1)k}\cdot\tilde{O}(n)$, and for \cov\ 
    over $(\Z/m\Z)^d$ with running time $\epsilon^{-2}d^{m-1}\cdot\tilde{O}(n)$.
	}
	\begin{thm}\label{thm:ov-algo}
		\stateovalgo{}
	\end{thm}
	
  If $\epsilon^{-1}$ and $d$ are at most polylogarithmic in~$n$, and $m$ is constant, these algorithms run in quasilinear time. Note that under SETH, any exact counting algorithm for \cov\ over $(\Z/m\Z)^d$ requires time $\Omega(n^{2-o(1)})$~\cite{williams-ov-hard}; we have therefore proved a separation between approximate and exact counting. (As an aside, this implies that the factor of $\epsilon^{-2}$ in the running time of Theorem~\ref{thm:finegrain} cannot be dropped to $\epsilon^{-1/2+o(1)}$ under SETH.)
  Williams and Yu showed that their algorithm's dependence on~$d$ is close to best possible under SETH, and this hardness result of course applies to approximate counting as well. 
    
	\subsubsection{3SUM}\label{sec:tsum-intro}

  In the \dtsum\ problem, we are given three integer lists $A$, $B$, and $C$ of total length $n$ and must decide whether there exists a tuple $(a,b,c) \in A\times B\times C$ with $a+b=c$.
  One popular extension is 3SUM+, due to Vassilevska~Williams and Williams~\cite{WW-NWT}, which asks for \dtsum\ to be solved for all inputs $(A,B,c)$ with $c \in C$. However, as we are specifically concerned with counting problems, we instead consider the problem \ctsum, where we must compute the total number of solution tuples $(a,b,c)$.
    
  It is easy to see that \dtsum\ and \ctsum\ can be solved in $\tilde{O}(n^2)$ operations by sorting~$C$ and iterating over all pairs in $A \times B$, and it is conjectured~\cite{GO-3sum,Patrascu-3sum} that \dtsum\ admits no $O(n^{2-\delta})$-time randomised algorithm for any constant $\delta>0$.
  This approach is also how we model instances of \dtsum\ as a bipartite graph in order to do approximate counting. Joining two vertices $a$ and $b$ whenever $a+b \in C$, adjacency queries can be answered efficiently by binary search on the now-sorted list~$C$, and independence queries on a set $S\subseteq A\cup B$ can be answered by the assumed decision algorithm.
  Analogous to Theorem~\ref{thm:ov}, in Section~\ref{sec:3sum-proofs} we obtain the following structural result for \dtsum\ as a corollary to Theorem~\ref{thm:finegrain}.

	\newcommand{\statetsum}{
		If $\dtsum$ with~$n$ integers has a randomised algorithm that runs in time $T(n)$, then there is a randomised $\epsilon$-approximation algorithm for $\ctsum$ that runs in time $T(n)\cdot \epsilon^{-2}O(\log^6 n\log\log n)$.
	}
	\begin{thm}\label{thm:tsum}
		\statetsum{}
	\end{thm}
	
	Thus if $\epsilon^{-1}$ is at most polylogarithmic in~$n$, then the approximate counting algorithm in Theorem~\ref{thm:tsum} has only polylogarithmic overhead over decision. Independently of whether or not the \dtsum\ conjecture is true, we conclude that \dtsum\ and, say, $\frac 12$-approximating \ctsum\ have the same time complexity up to polylogarithmic factors.
	
	The fastest-known algorithm for 3SUM, due to Baran, Demaine and P\v{a}tra\c{s}cu~\cite{Baran2008}, has running time $O(n^2(\log\log n/\log n)^2)$. Theorem~\ref{thm:tsum} does not currently yield improved algorithms for approximating \ctsum\ as the polylogarithmic speedup factor of~$o(\log^3 n)$ over exhaustive search is smaller than the $O(\log^6 n\log\log n)$ cost in our reduction. However, Chan and Lewenstein~\cite{CL-3SUM} prove that \dtsum\ has much faster algorithms when the input is restricted to instances in which elements of one list are somewhat clustered, in a sense made explicit below. (Their algorithm also works for 3SUM+, but not for \ctsum\ as far as we can tell.) This is an interesting special case with several applications, including monotone multi-dimensional \dtsum\ with linearly-bounded coordinates --- see the introduction of~\cite{CL-3SUM} for an overview. 
  Thus by using the algorithm of Chan and Lewenstein as a black box for the independence oracle, we obtain the following algorithm as a corollary to Theorem~\ref{thm:finegrain}.
	
	\newcommand{\statetsumalgo}{
		For all $\delta > 0$, there is a randomised $\epsilon$-approximation algorithm with running time $\epsilon^{-2}\cdot \tilde{O}(n^{2-\delta/7})$ for instances of \ctsum\ with $n$ integers such that at least one of $A$, $B$, or $C$ may be covered by $n^{1-\delta}$ intervals of length $n$.
	}
	\begin{thm}\label{thm:tsum-algo}
		\statetsumalgo{}
	\end{thm}
	
  \subsubsection{NWT}
  In the Negative-Weight-Triangle problem, we are given an edge-weighted graph and must decide whether the graph contains a triangle of negative total weight. Vassilevska~Williams and Williams~\cite{WW-NWT} prove that NWT is equivalent to APSP under subcubic reductions. An $n$-vertex instance of \dnwt\ and its natural counting version \cnwt\ can be solved in time $O(n^3)$ by exhaustively checking every possible triangle, and it is conjectured~\cite{WW-NWT} that \dnwt\ admits no $O(n^{3-\delta})$-time randomised algorithm for any constant~$\delta>0$.
  
  To reduce approximate \#NWT to its decision version \dnwt, we put all vertices on one side of the bipartite graph and all edges on the other side. Then adjacency queries correspond to testing whether a given vertex and edge together form a triangle of negative weight, and independence queries can be answered by the assumed decision algorithm for~\dnwt.
  Thus in Section~\ref{sec:nwt-proofs} we obtain the following structural result for NWT as a corollary to Theorem~\ref{thm:finegrain}.
	
	\newcommand{\statenwt}{
		If $\dnwt$ for $n$-vertex graphs has a randomised algorithm that runs in time $T(n)$, then there is a randomised
    $\epsilon$-approximation algorithm for $\cnwt$ that runs in time \[T(n)\cdot\epsilon^{-2}O(\log^6 n\log\log n)\,.\]
	}
	\begin{thm}\label{thm:nwt}
		\statenwt{}
	\end{thm}
	
	Thus if~$\epsilon^{-1}$ is at most polylogarithmic, our algorithm has only polylogarithmic overhead over decision.
  It is known~\cite{WW-NWT} that a truly subcubic algorithm for NWT implies that the negative-weight triangles can also be enumerated in subcubic time. While an enumeration algorithm is obviously stronger than an approximate counting algorithm, this reduction has polynomial overhead and so does not imply Theorem~\ref{thm:nwt}.
  
  Williams~\cite{Williams-APSP} gives an algorithm with subpolynomial improvements over the exhaustive search algorithm.
  Using this algorithm as a black-box to answer independence queries, we obtain the following algorithm as a corollary to Theorem~\ref{thm:finegrain}.
	\newcommand{\statenwtalgo}{
		There is a randomised $\epsilon$-approximation algorithm for \cnwt\ which runs in time $\epsilon^{-2}n^3/e^{\Omega(\sqrt{\log n})}$ on graphs with $n$ vertices and polynomially bounded edge-weights.
	}
	\begin{thm}\label{thm:nwt-algo}
		\statenwtalgo{}
	\end{thm}

	\subsection{Our results for the satisfiability problem}\label{sec:sat-intro}
  
  In $k$-SAT we are given a $k$-CNF formula with~$n$ variables and must decide whether it is satisfiable.
  In the natural counting version \#$k$-SAT, we must compute the number of satisfying assignments.
  The phenomenon that decision, approximate counting, and exact counting seem to become progressively more difficult is nicely represented in the literature:
  The most efficient known $3$-SAT algorithms run in time~$O(1.308^n)$ for decision (Hertli~\cite{DBLP:journals/siamcomp/Hertli14}), in time~$O(1.515^n)$ for $\frac{1}{2}$-approximate counting (Schmitt and Wanka~\cite{wanka-ksat}), and in time~$O(1.642^n)$ for exact counting (Kutzkov~\cite{DBLP:journals/ipl/Kutzkov07}).
  
  Schmitt and Wanka's algorithm is based on an approach of Thurley~\cite{DBLP:conf/stacs/Thurley12}. They reduce approximate counting to decision in such a way that an $\pf{O}(2^{(1-\delta_k)n})$-time algorithm for $k$-SAT is turned into an $\epsilon$-approximation algorithm for \#$k$-SAT that runs in time $\epsilon^{-2}\cdot\pf{O}(2^{(1-\delta_k')n})$ for some $\delta_k/2 < \delta_k' < \delta_k$. In the most general form of their algorithm, $\delta_k'$ depends on a complicated parameterisation and is calculated on an ad hoc basis for $k=3$ and $k=4$, so no asymptotics of $\delta_k'-\delta_k$ are available; the slightly weaker form given in Section 4 of their paper satisfies $\delta_k' \to \delta_k/2$ as $k\to\infty$. Thus the exponential savings over exhaustive search go down from~$\delta_k$ for decision to roughly $\delta_k/2$ for approximate counting. For example, in the extreme case that Impagliazzo and Paturi's~\cite{IP-ETH} exponential time hypothesis (ETH) is false and $3$-SAT can be solved in time~$2^{o(n)}$, their reduction would only yield an exponential-time algorithm for \#3-SAT.
  
  Traxler~\cite{Traxler} constructs a reduction from approximate counting to decision, in which savings of~$\delta$ for decision become $\delta-o(1)$ for approximate counting, so by this metric the reduction is efficient. However, this reduction creates clauses of width $\Omega(\log n)$ and so is not suitable for $k$-SAT when $k$ is a constant.
  
  We adapt the Valiant--Vazirani style approach of Calabro, Impagliazzo, Kabanets, and Paturi~\cite{CIKP} to obtain a reduction from approximate \#$k$-SAT to $k'$-SAT, with a trade-off between keeping $k'$ close to $k$ versus keeping the cost of the reduction low. At the extremes, writing $n$ for the number of variables in the \#$k$-SAT instance, it implies a reduction from approximate \#$k$-SAT to $k$-SAT with exponential overhead $2^{O(\log^2 k/k)n}$, or a reduction from approximate \#$k$-SAT to $O(k\log^2 k)$-SAT with subexponential overhead. We formally state this reduction as Theorem~\ref{thm:growthrate} in Section~\ref{sec:CNF}.
  
  Our reduction yields interesting structural corollaries for ETH and SETH.
  Recall that SETH is false if and only if there exists some $\delta>0$ such that $k$-SAT can be solved in time $O(2^{(1-\delta)n})$ for all constants~$k$.
  Our reduction implies not only that SETH is equivalent to its approximate counting version (which is also implied by~\cite{wanka-ksat} and~\cite{DBLP:conf/stacs/Thurley12}), but also that the exponential savings~$\delta$ must be the same:
  
  \newcommand{\stateapproxseth}{
		Let $0<\delta<1$. Suppose that for all $k\in\N$, there is a randomised algorithm which runs on $n$-variable instances of \cnfsat{k} in time $O(2^{(1-\delta)n})$. Then for all $\delta' > 0$ and all $k\in\N$, there is a randomised $\epsilon$-approximation algorithm which runs on $n$-variable instances of \ccnfsat{k} in time $\epsilon^{-2}\cdot O(2^{(1-\delta+\delta')n})$.
	}
	\begin{thm}\label{thm:approx-SETH}
		\stateapproxseth{}
	\end{thm}
  
  By the sparsification lemma~\cite{IPZ-spars}, ETH is false if and only if $k$-SAT can be solved in time $O(2^{\delta n})$ for all constant~$\delta>0$ and $k$.
  Since approximate counting always implies decision, ETH clearly implies its seemingly-weaker approximate counting formulation.
  By letting~$\delta$ increase to~$1$ in Theorem~\ref{thm:approx-SETH}, we see that the converse is also true:

	\newcommand{\stateapproxeth}{
	  ETH is false if and only if, for every $k \in \N$ and $\delta > 0$, there is a randomised $\epsilon$-approximation algorithm that runs  on $n$-variable instances of \ccnfsat{k} in time~${\epsilon^{-2}\cdot O(2^{\delta n})}$.
	}
	\begin{thm}\label{thm:approx-ETH}
		\stateapproxeth{}
	\end{thm}
	
	It remains an open and interesting question whether a result analogous to Theorem~\ref{thm:approx-SETH} holds for fixed $k$, that is, whether deciding $\cnfsat{k}$ and approximating $\ccnfsat{k}$ have the same time complexity up to a subexponential factor. Even a small improvement on Theorem~\ref{thm:growthrate} would lead to new algorithms for approximate \#$k$-SAT. Indeed, for large constant $k$, the best-known decision, $\frac{1}{2}$-approximate counting, and exact counting algorithms (due to Paturi, Pudl\'{a}k, Saks, and Zane~\cite{PPSZ}, Schmitt and Wanka~\cite{wanka-ksat}, and 
	Impagliazzo, Matthews, and Paturi~\cite{DBLP:conf/soda/ImpagliazzoMP12},
	respectively) all have running time $2^{(1-\Theta(1/k))n}$, but with progressively worse constants in the exponent. 
	If our reduction from approximate \#$k$-SAT to $k$-SAT could be improved so that the exponential overhead were $2^{o(1/k)}$ instead of $2^{O((\log k)^2/k)}$, this would yield faster approximate counting algorithms for large but constant~$k$.

\subsection{Techniques}\label{sec:methods}

Our techniques for the CNF-SAT and the fine-grained results are independent from each other.

\paragraph{CNF-SAT results.}
We first discuss Theorems~\ref{thm:approx-SETH} and~\ref{thm:approx-ETH}, which we prove in Section~\ref{sec:CNF}. In the polynomial setting, the standard reduction from approximating \ccnfsat{k} to deciding \cnfsat{k} is due to Valiant and Vazirani~\cite{VV}, and runs as follows. If a $k$-CNF formula $F$ has at most $2^{\delta n}$ solutions for some~$\delta > 0$, then we use a standard branching algorithm with~$\pf{O}(2^{\delta n})$ calls to a \cnfsat{k}-oracle to prune the search tree to size $O(2^{\delta n})$. Otherwise $F$ has many solutions, and for any $m \in \N$, one may form a new formula $F_m$ by conjoining $F$ with $m$ independently-chosen uniformly random XOR clauses. It is relatively easy to see that as long as the number $\SAT{F}$ of satisfying assignments of~$F$ is substantially greater than $2^m$, then $\SAT{F_m}$ is concentrated around $2^{-m}\SAT{F}$. By choosing~$m$ appropriately, one may reduce $\SAT{F_m}$ to below $2^{\delta n}$ and thus compute $\SAT{F_m}$ exactly, then multiply it by $2^m$ to obtain an estimate for $\SAT{F}$. 

Unfortunately, this argument requires modification in the exponential setting. If $F$ has $n$ variables, then each uniformly random XOR has length $\Theta(n)$ and therefore cannot be expressed as a \mbox{$k$-CNF} formula without introducing $\Omega(n)$ new variables. It follows that (for example) $F_{\floor{n/2}}$ will contain~$\Theta(n^2)$ variables. This blowup is acceptable in a polynomial setting, but not an exponential one --- for example, given a $\Theta(2^{n^{2/3}})$-time algorithm for \cnfsat{k}, it would yield a useless $\Theta(2^{n^{4/3}})$-time randomised approximate counting algorithm for \ccnfsat{k}. We can afford to add only constant-length~XORs, which do not in general result in concentration in the number of solutions.

We therefore make use of a hashing scheme developed by Calabro, Impagliazzo, Kabanets, and Paturi~\cite{CIKP} for a related problem, that of reducing \cnfsat{k} to Unique-\cnfsat{k}. They choose a $2s$-sized subset of $[n]$ uniformly at random, where $s$ is a large constant, then choose variables independently at random within that set. This still does not yield concentration in the number of solutions of~$F_m$, but it turns out that the variance is sufficiently low that we can remedy this by summing over many slightly stronger independently-chosen hashes. 

\paragraph{Fine-grained results.}
We now sketch the proof of Theorem~\ref{thm:finegrain}, which we prove in Section~\ref{sec:finegrain}. Given a bipartite graph with $G=(U,V,E)$ and $X \subseteq V$, we write $\eb{X}$ for the number of edges incident to $X$. For all $X \subseteq V$, we may halve $\eb{X}$ in expectation simply by removing half the vertices in $X$ chosen independently at random. Moreover, if $\eb{X}$ is sufficiently small, we may use binary search to efficiently determine $\eb{X}$ exactly. Thus, as with Theorems~\ref{thm:approx-SETH} and~\ref{thm:approx-ETH}, we might hope to implement the classical approach of Valiant and Vazirani~\cite{VV}; start with $X = V$ (so that $\eb{X} = e(G)$), repeatedly approximately halve $\eb{X}$ until it is small enough to determine exactly, then multiply by the appropriate power of $2$ and output the result. 

Unfortunately, this naive algorithm may fail. For example, if the non-isolated vertices of $G$ form a star whose central vertex lies in $V$, then the new value of $\eb{X}$ is clearly not concentrated around its expectation; it is either unchanged or reduced to zero. In Lemma~\ref{lem:balanced-halve}, we show using martingale techniques that this is essentially the only way things can go wrong. We say $X$ is \emph{balanced} if no single vertex in $X$ is incident to a large proportion of the edges in $G[U\cup X]$ (see Definition~\ref{defn:balanced}), and Lemma~\ref{lem:balanced-halve} shows that if $X$ is balanced then with high probability we can approximately halve~$\eb{X}$ by deleting half of $X$ uniformly at random. 

We therefore proceed by finding a small set of vertices which ``unbalances'' $X$ if one exists, approximately counting the edges incident to them, and removing them from $X$. We repeat this process as necessary until $X$ becomes balanced, then delete half of what remains. At the end, we approximate $e(G)$ by taking an appropriate linear combination of our edge counts at each stage. However, since our access to the graph is limited, it is non-trivial to find the ``unbalancing'' vertices. We must also show that we do not remove too many vertices in this way, as finding edges by brute force is computationally expensive. Our algorithm is essentially given by \counter{} on p.~\pageref{algo:edgecount-new}, with some trivial modifications as described in the proof of Theorem~\ref{thm:finegrain}.

\section{Preliminaries}\label{sec:prelim}

\subsection{Notation}

We write $\N$ for the set of all positive integers. For a positive integer~$n$, we use $[n]$ to denote the set $\{1,\dots,n\}$. We use $\log$ or $\ln$ to denote the base-$e$ logarithm, and $\lg$ to denote the base-$2$ logarithm.

We consider graphs~$G$ to be undirected, and write $e(G) = |E(G)|$. For all $v \in V(G)$, we use $N(v)$ to denote the neighbourhood $\setc{w \in V(G)}{\set{v,w} \in E(G)}$ of $v$. 
For all $X \subseteq V(G)$, we define $\nbh{X} = \bigcup_{v \in X} N(v)$. We define $\eb{X}$ to be the size of the edge boundary of~$X$, that is, $\eb{X} = |\{e \in E(G) \mid |e \cap X| = 1\}|$.
For convenience, we shall generally present bipartite graphs~$G$ as a triple $(U,V,E)$ in which $(U,V)$ is a partition of $V(G)$ and $E \subseteq U\times V$.

When stating quantitative bounds on running times of algorithms, we assume the standard word-RAM machine model with logarithmic-sized words. We assume that lists and functions in the problem input are presented in the natural way, that is, as an array using at least one word per entry, and we assume that numerical values such as the edge weights in NWT are given in binary.
We shall write $f(x) = \tilde{O}(g(x))$ when for some constant $c \in \R$, $f(x) = O((\log x)^cg(x))$ as $x\rightarrow\infty$. Similarly, we write $f(x) = \pf{O}(g(x))$ when for some constant $c \in \R$, $f(x) = O(x^cg(x))$ as $x\rightarrow\infty$.

We require our problem inputs to be given as finite binary strings, and write $\Sigma^*$ for the set of all such strings. A \emph{randomised approximation scheme} for a function $f:\Sigma^*\rightarrow\mathbb{N}$ is a randomised algorithm that takes as input an instance $x \in \Sigma^*$ and a rational error tolerance $0 < \epsilon < 1$, and outputs a rational number $z$ (a random variable depending on the ``coin tosses'' made by the algorithm) such that, for every instance $x$, $\pr((1-\epsilon)f(x) \le z \le (1+\epsilon)f(x)) \ge 2/3$. All of our approximate counting algorithms will be randomised approximation schemes.

\subsection{Probability theory}

We use some results from probability theory, which we collate here for reference. First, we state Chebyshev's inequality.

\begin{lem}\label{lem:chebyshev}
  Let $X$ be a real-valued random variable with mean $\mu$ and let $t > 0$.
  Then
  \[
      \pushQED{\qed}
      \pr\paren[\Big]{|X - \mu| \ge t} \le \frac{\var(X)}{t^2}\,.\qedhere\popQED
  \]
\end{lem}

We also use the following concentration result due to McDiarmid~\cite{mcdiarmid}.

\begin{lem}\label{lem:mcdiarmid}
	Let $f$ be a real function of independent random variables $X_1, \dots, X_m$, and let $\mu = \E(f(X_1, \dots, X_m))$. Let $c_1, \dots, c_m \ge 0$ such that, for all $i \in [m]$ and all pairs $(\vec{x},\vec{x'})$ differing only in the $i$th coordinate, we have $|f(\vec{x})-f(\vec{x'})| \le c_i$. Then for all $t > 0$, 
	\[
		\pushQED{\qed}
		\pr(|f(X_1, \dots, X_m) - \mu| \ge t) \le 2e^{-2t^2/\sum_{i=1}^m c_i^2}.\qedhere\popQED
	\]
\end{lem}

Finally, we use the following Chernoff bounds, proved in (for example) Corollaries 2.3--2.4 and Remark 2.11 of Janson, {\L}uczak and Rucinski~\cite{JLR}.

\begin{lem}\label{lem:chernoff}
  Let $X$ be a binomial or hypergeometric random variable with mean $\mu$.
	\begin{enumerate}[label=(\roman*)]
		\item For all $\epsilon$ with $0 < \epsilon \le\tfrac{3}{2}$, we have $\pr(|X - \mu| \ge \epsilon\mu) \le 2e^{-\epsilon^2\mu/3}$.
		\item For all $t$ with $t \ge 7\mu$, we have $\pr(X \ge t) \le e^{-t}$.\qed{}
	\end{enumerate}
\end{lem}

\section{From decision to approximate counting CNF-SAT}\label{sec:CNF}
In this section we prove our results for the satisfiability of CNF formulae, formally defined as follows.

\defproblem{\cnfsat{k}} 
{A $k$-CNF formula $F$.}
{Decide if~$F$ is satisfiable.} 

\defproblem{\ccnfsat{k}}
{A $k$-CNF formula $F$.}
{Compute the number $\SAT{F}$ of satisfying assignments of~$F$.}

We also define a technical intermediate problem. For all $s \in \N$, we say that a matrix $A$ is \emph{$s$-sparse} if every row of $A$ contains at most $s$ non-zero entries. In the following definition, $k\in\N$ and $s\in\N$ are constants. 

\defproblem{$\Pi_{k,s}$}
{An $n$-variable Boolean formula $F$ of the form $F(\vec{x}) = F'(\vec{x}) \wedge (A\vec{x} = \vec{b})$. Here $F'$ is a $k$-CNF formula, $A$ is an $s$-sparse $m\times n$ matrix over $\GF{2}$ with $0 \le m \le n$, and $\vec{b} \in \GF{2}^m$.}
{Decide if $F$ is satisfiable.}

We define the \emph{growth rate $\pi_{k,s}$} of $\Pi_{k,s}$ as the infimum over all $\beta>0$ such that $\Pi_{k,s}$ has a randomised algorithm that runs in time $\pf{O}(2^{\beta n})$ and outputs the correct answer with probability at least $2/3$.
Our main reduction is encapsulated in the following theorem.

\newcommand{\stategrowth}{
  Let $k \in \N$ with $k \ge 2$, let $0 < \delta < 1$, and let $s \ge 120\lg^2(6/\delta)/\delta$. Then there is a randomised approximation scheme for \ccnfsat{k} which, when given an $n$-variable formula $F$ and approximation error parameter $\epsilon$, runs in time $\epsilon^{-2} \cdot O\paren[\big]{2^{(\pi_{k,s}+\delta)n}}$.
}
\begin{thm}\label{thm:growthrate}
	\stategrowth{}
\end{thm}

Before we prove this theorem, let us derive Theorems~\ref{thm:approx-SETH} and~\ref{thm:approx-ETH} as immediate corollaries.
In both cases, we use the fact that the condition $A\vec{x}=\vec{b}$ can be expressed as an~$s$-CNF formula with~$m 2^{s-1}$ clauses, and thus $\pi_{k,s}\le\pi_{\max\set{k,s},0}$ holds for all constant~$k,s$.

\begin{repthm}{thm:approx-SETH}
  \stateapproxseth{}
\end{repthm}
\begin{proof}
	Let $\delta > 0$ be as specified in the theorem statement. Then for all constant $k,s\in\N$, we have $\pi_{k,s} \le \pi_{\max\set{k,s},0}\le 1-\delta$. The result follows by Theorem~\ref{thm:growthrate} with $s=120\lg^2(6/\delta')/\delta'$.
\end{proof}

\begin{repthm}{thm:approx-ETH}
  \stateapproxeth{}
\end{repthm}

\begin{proof}
  The backward implication is immediate: Any randomised $\frac12$-approximation scheme for \ccnfsat{3} is able to decide \cnfsat{3} with success probability at least 2/3.
  For the forward implication, assume ETH is false.
  By the sparsification lemma~\cite[Lemma 10]{IPZ-spars}, we then have $\pi_{k,0}=0$ for all~$k\in\N$. Hence for all $k,s\in\N$, we obtain $\pi_{k,s}\le \pi_{\max\set{k,s},0}= 0$. The result now follows by~Theorem~\ref{thm:growthrate}.
\end{proof}

\subsection{Proof of Theorem~\ref{thm:growthrate}} 

Given access to an oracle that decides satisfiability queries, we can compute the exact number of solutions of a formula with few solutions using a standard self-reducibility argument given below (see also~\cite[Lemma~3.2]{DBLP:conf/stacs/Thurley12}).

\newcommand{\satfew}[1]{{\normalfont\texttt{CountFew}}} \label{algo:sparse}
\begin{algorr}{$\satfew{\delta}(F,a)$}{Given an instance $F$ of $\Pi_{k,s}$ on $n$ variables, $a \in \N$, and access to an oracle for $\Pi_{k,s}$, this algorithm computes $\SAT{F}$ if $\SAT{F} \le a$; otherwise it outputs \textnormal{FAIL}.}
	\item\label{countfew:oracle} \emph{(Query the oracle)} If $F$ is unsatisfiable, return $0$.
	\item\label{countfew:novars} \emph{(No variables left)} If $F$ contains no variables, return~$1$.
	\item\label{countfew:branch} \emph{(Branch and recurse)} Let $F_0$ and $F_1$ be the formulae obtained from $F$ by setting the first free variable in~$F$ to 0 and 1, respectively.
	  If $\satfew{\delta}(F_0,a)+\satfew{\delta}(F_1,a)$ is at most~$a$, then return this sum; otherwise abort the entire computation and return FAIL.
\end{algorr}

\begin{lem}\label{lem:smallcount}
	$\satfew{\delta}$ is correct and runs in time at most $(\min\set{a,\SAT{F}}+1)\cdot \tilde O(|F|)$. Moreover, each oracle query is a formula with at most $n$ variables.
\end{lem}
\begin{proof}
  The correctness of $\satfew{\delta}$ follows by induction from $\SAT{F}=\SAT{F_0}+\SAT{F_1}$.
  For the running time, consider the recursion tree of $\satfew{\delta}$ on inputs~$F$ and~$a$.
  At each vertex, the algorithm takes time at most~$\tilde O(\abs{F})$ to compute $F_0$ and $F_1$, and it issues a single oracle call.
  For convenience, we call the leaves of the tree at which $\satfew{\delta}$ returns 0 in Step~\ref{countfew:oracle} or 1 in Step~\ref{countfew:novars} the \emph{0-leaves} and \emph{1-leaves}, respectively.
  Let $x$ be the number of 1-leaves. Each non-leaf is on the path from some 1-leaf to the root, otherwise it would be a 0-leaf. There are at most $x$ such paths, so there are at most $nx$ non-leaf vertices in total. Finally, every 0-leaf has a sibling which is not a 0-leaf, or its parent would be a 0-leaf, so there are at most $(n+1)x$ 0-leaves in total. Overall, the tree has at most $4nx$ vertices. An easy induction using Step~\ref{countfew:branch} implies that $x \le 2a$, and certainly $x \le \SAT{F}$, so the claimed running time is correct.
\end{proof}

When our input formula $F$ has too many solutions to apply $\satfew{\delta}$ efficiently, we first reduce the number of solutions by hashing.
In particular, we use the same hash functions as Calabro et al.~\cite{CIKP};
they are based on random sparse matrices over $\GF{2}$ and formally defined as follows:

\begin{defn}\label{def:CIKP}
  Let $s,m,n \in \N$.
  An \emph{$(s,m,n)$-hash} is a random $m\times n$ matrix $A$ over $\GF{2}$ defined as follows.
  For each row $i \in [m]$, let $R_i$ be a uniformly random size-$s$ subset of $[n]$.
  Then for all $i \in [m]$ and all $j \in R_i$, we choose values $A_{i,j}\in\GF{2}$ independently and uniformly at random, and set all other entries of $A$ to zero.
\end{defn}

For intuition, suppose that $F$ is an $n$-variable $k$-CNF formula, $S$ is the set of satisfying assignments of $F$, and $|S| > 2^{\delta n}$ holds for some small $\delta>0$. It is easy to see that, for all $m,s \in \N$ and uniformly random $\vec{b} \in \GF{2}^m$, if $A$ is an $(s,m,n)$-hash, then the number~$X$ of satisfying assignments of $F(\vec{x}) \wedge (A\vec{x} = \vec{b})$ has expected value $|S|/2^m$. (See Lemma~\ref{lem:inaccurate-hash}.) If~$X$ were concentrated around its expectation, then by choosing an appropriate value of $m$, we could reduce the number of solutions to at most $2^{\delta n}$, apply $\satfew{\delta}$ to count them exactly, then multiply the result by $2^m$ to obtain an approximation to $|S|$. This is the usual approach pioneered by Valiant and Vazirani~\cite{VV}. 

In the exponential setting, however, we can only afford to take $s = O(1)$, which means that~$X$ is not in general concentrated around its expectation. In~\cite{CIKP}, only very limited concentration was needed, but we require strong concentration. To achieve this, rather than counting satisfying assignments of a single formula $F(\vec{x}) \wedge (A\vec{x} = \vec{b})$, we will sum over many such formulae.
We first bound the variance of an individual $(s,m,n)$-hash when $s$ and $S$ are suitably large. Our analysis here is similar to that of Calabro et al.~\cite{CIKP}, although they are concerned with lower-bounding the probability that at least one solution remains after hashing and do not give bounds on variance.

\begin{lem}\label{lem:inaccurate-hash}
  Let $\delta\in\R$ with $0 < \delta < \frac16$ and let $s,m,n\in\N$. Suppose $m \le n$ and $s \ge 20\lg^2(1/\delta)/\delta$. Let $S \subseteq \GF{2}^n$ and suppose $|S| \ge 2^{m+\delta n}$.
  Let $A$ be an $(s,m,n)$-hash, and let $\vec{b} \in \GF{2}^m$ be uniformly random and independent of $A$.
  Let $S' = \setc{\vec{x} \in S }{ A\vec{x} = \vec{b}}$.
  Then $\E(|S'|) = 2^{-m}|S|$ and $\var(|S'|) \le |S|^2 2^{\delta n/8-2m}$.
\end{lem}

	\begin{proof}
		For each $\vec{x} \in \GF{2}$, let $I_{\vec{x}}$ be the indicator variable of the event $A\vec{x}=\vec{b}$. Exposing $A$ implies $\pr(I_{\vec{x}}) = 2^{-m}$ for all $\vec{x} \in \GF{2}^n$, and hence
		\[\E(|S'|) = \sum_{\vec{x} \in S}\pr(I_{\vec{x}}) = 2^{-m}|S|.\]
		
		We now bound the second moment. We have
		\begin{align}\nonumber
			\E(|S'|^2) 
			&= \sum_{(\vec{x},\vec{y}) \in S^2} \E(I_{\vec{x}}I_{\vec{y}}) 
			= \sum_{(\vec{x},\vec{y}) \in S^2} \pr(A\vec{x} = A\vec{y} = \vec{b})\\\label{eqn:inacc-hash-1}
			&= \sum_{(\vec{x},\vec{y}) \in S^2} \prod_{i=1}^m \pr((A\vec{x})_i = (A\vec{y})_i = \vec{b}_i).
		\end{align}
    When $\vec{x}$ and $\vec{y}$ are fixed, the events in~\eqref{eqn:inacc-hash-1} are identically distributed and we write
    $p_{\vec x,\vec y}=\pr(\vec{a}^T\vec{x} = \vec{a}^T\vec{y} = b)$, where $b\in\set{0,1}$ is sampled uniformly at random and $\vec{a}\in\set{0,1}^n$ is sampled by first sampling a size-$s$ set $R\subseteq\set{1,\dots,n}$ and then setting the bits $\vec{a}_j$ uniformly for~$j\in R$, and $\vec{a}_j=0$ for $j\not\in R$.
    Using this shorthand notation, we split the sum in~\eqref{eqn:inacc-hash-1} depending on whether the Hamming distance~$d(\vec{x},\vec{y})$ between the vectors is at most~$\alpha n$ or larger, for some parameter~$\alpha<\frac12$ specified later.
    \begin{align}\label{eqn:inacc-hash-2}
			\E(|S'|^2) 
			&= \sum_{(\vec{x},\vec{y}) \in S^2} p_{\vec{x},\vec{y}}^m
      =
       \sum_{\substack{(\vec{x},\vec{y}) \in S^2 \\ d(\vec{x},\vec{y}) \le \alpha n}}
       p_{\vec{x},\vec{y}}^m
       +
       \sum_{\substack{(\vec{x},\vec{y}) \in S^2 \\ d(\vec{x},\vec{y}) > \alpha n}}
       p_{\vec{x},\vec{y}}^m\,.
    \end{align}
    We now provide upper bounds for these two sums.
    For the first sum, let us write $h:[0,1]\rightarrow[0,1]$ for the binary entropy function $h(\alpha)=-\alpha \lg \alpha - (1-\alpha)\lg(1-\alpha)$; it is known that the Hamming ball of radius~$\alpha n$ around a binary vector~$\vec{x}$ contains at most~$2^{h(\alpha) n}$ binary vectors~$\vec{y}$.
    Thus the first sum is bounded by~$\abs{S} 2^{h(\alpha)n} \max\set{p_{\vec{x},\vec{y}}^m}$.
    To bound the maximum, note by exposing~$\vec{a}$ that~$p_{\vec{x},\vec{y}}\le\frac 12$ holds for all~$\vec{x},\vec{y}$.
    Thus, the first sum in~\eqref{eqn:inacc-hash-2} is bounded by~$\abs{S}2^{h(\alpha)n-m}$.
    
    The second sum in~\eqref{eqn:inacc-hash-2} is at most
    $\abs{S}^2 \max\setc{p_{\vec{x},\vec{y}}^m}{d(\vec x,\vec y)>\alpha n}$, and so it remains to bound $p_{\vec{x},\vec{y}}$ for vectors~$\vec x$ and $\vec y$ whose distance is more than~$\alpha n$.
    Write $\vec{x}_R\in\GF{2}^R$ for the projection of~$\vec{x}$ to the coordinates of~$R$.
    Conditioning on the event~$\vec x_R=\vec y_R$, we get
    \begin{align}\nonumber
      p_{\vec x,\vec y}
      &=
      \pr\paren[\Big]{
      \vec{a}^T\vec{x}=\vec{a}^T\vec{y}=b
      \;\big\vert\;\vec x_R\ne\vec y_R
      }\cdot
      \pr(\vec x_R\ne\vec y_R)\\\nonumber
      &\qquad+
      \pr\paren[\Big]{
      \vec{a}^T\vec{x}=\vec{a}^T\vec{y}=b
      \;\big\vert\;\vec x_R=\vec y_R
      }\cdot
      \pr(\vec x_R=\vec y_R)
      \\\label{eq:pxy}
      &\le
      \pr\paren[\Big]{
      \vec{a}^T\vec{x}=\vec{a}^T\vec{y}=b
      \;\big\vert\;\vec x_R\ne\vec y_R
      }
      +
      \tfrac12\cdot
      \pr\paren[\Big]{
      \vec x_R=\vec y_R
	  }
	  \,.
    \end{align}
    We claim that the first summand of~\eqref{eq:pxy} is equal to~$\frac14$ and the second is bounded above by $\frac12 e^{-\alpha s}$.
    Indeed, conditioned on $\vec x_R\ne \vec y_R$, there is a coordinate~$c\in R$ with $\vec x_c\ne\vec y_c$.
    Without loss of generality, assume $\vec x_c=1$ and $\vec y_c=0$.
    Under this conditioning, the events $\vec a^T\vec x=\vec a^T\vec y$ and $\vec a^T\vec y=b$ are actually independent, because~$\vec a_c$ is a uniform bit that only affects the first event and~$b$ is a uniform bit that only affects the second. More precisely, after exposing~$R$ with~$\vec x_R\ne \vec y_R$ and~$\vec a_j$ for all $j\in R\setminus\set{c}$, the probability that~$\vec a_c$ and~$b$ are set correctly is~$\frac 14$.
    To bound the second summand of~\eqref{eq:pxy}, recall that $d(\vec{x},\vec{y}) \ge \alpha n$ and $|R| = s$, and observe
		\begin{align*}
			\pr\paren[\Big]{\vec{x}_{R} = \vec{y}_{R}}
      &\le \frac{\binom{n-\lceil\alpha n\rceil}{s}}{\binom{n}{s}} 
			\le (1-\lceil \alpha n\rceil /n)^s
			\le e^{-\alpha s}.
		\end{align*}
    Putting the bounds on the terms in~\eqref{eq:pxy} together, we arrive at
    \begin{align*}
      p_{\vec{x},\vec{y}}
      \le\tfrac14+\tfrac12 e^{-\alpha s}
      =\tfrac14(1+2e^{-\alpha s})
      \le \tfrac14 e^{2e^{-\alpha s}}\,.
    \end{align*}
    This allows us to bound the second moment and thus the variance as well:
    \begin{align}\label{eq:varbound}
      \var(|S'|) = \E(|S'|^2)-\E(|S'|)^2
			&\le
      \paren[\Big]{
      \abs{S} 2^{h(\alpha)n-m}
      +
      \abs{S}^2 4^{-m} e^{m\cdot 2e^{-\alpha s}}
      }
      - \abs{S}^2 2^{-2m}
      \,.
    \end{align}
    By assumption we have~$\abs{S}\ge 2^{m+\delta n}$, and thus
    $\abs{S}^2 2^{-2m} \ge \abs{S} 2^{\delta n-m}$.
    Now we set $\alpha<\frac12$ such that $h(\alpha)=\delta$ holds.
    Since $\delta < \frac16$, we have $\alpha = h^{-1}(\delta) \ge \delta/(2\lg(6/\delta)) \ge \delta/(4\lg(1/\delta))$.
    It follows that
    $\alpha s \ge 5\lg(1/\delta) \ge 2\ln(4/\delta)$,
    and together with~\eqref{eq:varbound} we get $\var(|S'|) \le |S|^2e^{\delta^2 m/8}/2^{2m}$.
    Since $m \le n$ and $\delta < 1/\lg(e)$, the result follows.
	\end{proof}
	
	We now state our algorithm for Theorem~\ref{thm:growthrate} that reduces from approximate counting for $k$-SAT to decision for $\Pi_{k,s}$.
  In the following definition, $\delta$ is a rational constant with $0 < \delta < \frac13$.
	
	\newcommand{\satalgo}[1]{{\normalfont\texttt{ApxToD}}$_{#1}$}
	\begin{algorr}
		{\satalgo{\delta}}
    {Given an $n$-variable instance $F$ of \ccnfsat{k}, a rational number~$\epsilon\in(0,1)$, and access to an oracle for $\Pi_{k,s}$ for some 
    	$s \ge 40\lg^2(2/\delta)/\delta$,
		this algorithm computes a rational number $z$ 
		such that $(1-\epsilon)\SAT{F} \le z \le (1+\epsilon)\SAT{F}$ holds with probability at least $\frac34$.}
    \item\label{apxtod:constant} \emph{(Brute-force on constant-size instances)}\\
		If $n/\lg n \le 8/\delta$, solve the problem by brute force and return the result.
    \item\label{apxtod:few} \emph{(If there are few satisfying assignments, count them exactly)}\\
      Let $t = \ceil{\delta n/2+2\lg(1/\epsilon)}$, and apply $\satfew{}$ to $F$ and $a=2^{t+\delta n/2}$. Return the result if it is not equal to FAIL.
    \item\label{apxtod:outer} \emph{(Try larger and larger equation systems)}
      For each $m\in\set{0,\dots,n-t}$:
    \begin{enumerate}[label=\texttt{\alph*}]
      \item\label{apxtod:inner} For each $i\in\set{1,\dots,2^t}$:
        \begin{itemize}
          \item \emph{(Prepare query)}
            Independently sample an $(s,m+t,n)$-hash~$A_{m,i}$ and a uniformly random vector~$\vec{b_{m,i}} \in \GF{2}^{m+t}$.
            Let $F_{m,i}=F(\vec{x}) \wedge (A_{m,i}\vec{x} = \vec{b_{m,i}})$. 
          \item \emph{(Ask oracle using subroutine)}
            Let $z_{m,i}$ be the output of~$\satfew{}\paren*{F_{m,i},4a}$.
      \item \emph{(Bad randomness or $m$ too small)}
            If $z_{m,i}=\mbox{FAIL}$ or if $\sum_{j=1}^{i} z_{m,j}>4a$, then go to the next~$m$ in the outer for-loop.
            \end{itemize}
      \item\label{apxtod:return} \emph{(Return our estimate)}
        Return~$z=2^m\sum_{i=1}^{2^t} z_{m,i}$.
      \end{enumerate}
	\end{algorr}

	\begin{lem}\label{lem:satalgo}
    \satalgo{\delta} is correct for all $\delta \in(0, \frac13)$ and runs in time at most $\epsilon^{-2}\cdot \pf{O}(2^{\delta n})$.
    Moreover, the oracle is only called on instances with at most $n$ variables.
	\end{lem}
	\begin{proof}
    Let $F$ be a $k$-CNF formula on $n$ variables and let $\epsilon\in(0,1)$.
    For the running time, note that \aref{apxtod:constant} takes time~$O(2^{1/\delta})=O(1)$,  \aref{apxtod:few} takes time at most $\pf{O}(a)$ by Lemma~\ref{lem:smallcount}.
    By the same lemma, each invocation of $\satfew{}$ on input~$F_{m,i}$ in~\ref{apxtod:outer} takes time~$\pf{O}(\min\set{z_{m,i},a}+1)$.
    Moreover, the outer loop in \aref{apxtod:constant} is run at most~$n-t$ times, and for each  fixed~$m$, executing  \aref{apxtod:outer}\ref{apxtod:inner} in its entirety takes time at most $\pf{O}(a)$ due to the check whether $\sum_{j=1}^i z_{m,k}>4a$ holds.
    Thus the overall running time of the algorithm is~$\pf{O}(a)\le \pf{O}(\epsilon^{-2}2^{\delta n})$ as required.

    It remains to prove the correctness of the algorithm.
    If it terminates at \aref{apxtod:constant} or \aref{apxtod:few}, then correctness is immediate from Lemma~\ref{lem:smallcount}. Suppose not, so that $n/\lg n > 8/\delta$ holds, and the set $S$ of solutions of $F$ satisfies $|S| \ge 2^{t+\delta n/2}$. Let $M = \max\setc{m \in \Z }{ |S| \ge 2^{m+t+\delta n/2}}$, and note that $0 \le M \le n-t$ and $|S| \le 2^{M+t+\delta n/2+1}$.
    The formulas~$F_{m,i}$ are oblivious to the execution of the algorithm, so for the analysis we may view them as being sampled in advance.
    Let $S_{m,i}$ be the set of solutions to~$F_{m,i}$.
    For each~$m$ with $0 \le m \le M$, let $\mathcal{E}_m$ be the following event:
		\[
			\left|\sum_{i=1}^{2^t}|S_{m,i}| - 2^{-m}|S|\right| \le 2^{-m-(t-\delta n/2)/2}\cdot |S| \,.
		\]
	Thus $\mathcal{E}_m$ implies $\left|2^m\sum_{i=1}^{2^t}|S_{m,i}| - |S|\right| \le \epsilon|S|$.
	By Lemma~\ref{lem:inaccurate-hash} applied with
	$\delta/2$ in place of $\delta$ and $m+t$ in place of $m$, for all $0 \le m \le M$ and $1 \le i \le 2^t$, we have $\E(|S_{m,i}|) = 2^{-m-t}|S|$ and $\var(|S_{m,i}|) \le |S|^2 2^{\delta n/16-2m-2t}$.
	Since the $S_{m,i}$'s are independent, it follows by Lemma~\ref{lem:chebyshev} that 
		\[
			\pr(\mathcal{E}_m) \ge 1 - \frac{2^t \cdot |S|^22^{\delta n/16-2m-2t}}{2^{-2m-t+\delta n/2}|S|^2} \ge 1-2^{-\delta n/4} \ge 1-1/n^2.
		\]
		Thus a union bound implies that, with probability at least $3/4$, the event~$\mathcal{E}_m$ occurs for all~$m$ with $0 \le m \le M$ simultaneously.
		Suppose now that this happens.
    Then in particular, we have~\[\sum_{i=1}^{2^t}\abs{S_{M,i}}\le(1+\epsilon)2^{-M}\abs{S}\le 2^{t+\delta n /2 + 2}\,.\]
    But then, if \satalgo{\delta} reaches iteration~$m=M$, none of the calls to $\satfew{}$ fail in this iteration and we have $z_{M,i}=\abs{S_{M,i}}$ for all~$i\in\set{1,\dots,2^t}$.
    Thus \satalgo{\delta} returns some estimate~$z$ while $m\le M$.
    Moreover, since $\mathcal{E}_m$ occurs, this estimate satisfies
    $(1-\epsilon)|S| \le z \le (1+\epsilon)|S|$ as required.
   Thus \satalgo{\delta} behaves correctly with probability at least $3/4$, and the result follows.
	\end{proof}
	
  \begin{repthm}{thm:growthrate}
    \stategrowth{}
  \end{repthm}
	\begin{proof}
    If $\epsilon < 2^{-n}$, then we solve the \ccnfsat{k} instance exactly by brute force in time $\pf{O}(\epsilon^{-1})$, so suppose $\epsilon \ge 2^{-n}$. By the definition of $\pi_{k,s}$, there exists a randomised algorithm for $\Pi_{k,s}$ with failure probability at most $1/3$ and running time at most $\pf{O}(2^{(\pi_{k,s}+\delta/3)n})$. By Lemma~\ref{lem:chernoff}(i), for any constant $C$, by applying this algorithm $\lg(1/\epsilon)\cdot O(n) = O(n^2)$ times and outputting the majority answer, we may reduce the failure probability to at most $\epsilon^2/Cn2^{\delta n/3}$. We apply \satalgo{\delta/3} to $F$ and $\epsilon$, using the randomized algorithm for $\Pi_{k,s}$ in place of the $\Pi_{k,s}$-oracle. If we take~$C$ sufficiently large, then by Lemma~\ref{lem:satalgo} and a union bound, the overall failure probability is at most~$1/3$, and the running time is $\epsilon^{-2}\cdot\pf{O}(2^{(\pi_{k,s}+2\delta/3)n}) = \epsilon^{-2}\cdot O(2^{(\pi_{k,s}+\delta)n})$ as required.
	\end{proof}
	
	\section{Approximately Counting Edges in Bipartite Graphs}\label{sec:finegrain}
	
  In this section, we prove our main result, Theorem~\ref{thm:finegrain}.
  Recall from Section~\ref{sec:framework} that it consists of an algorithm that is given access to a bipartite graph via an adjacency oracle and an independence oracle. Throughout this section, we fix $G = (U,V,E)$ and $\epsilon > 0$ as the input to our edge-counting algorithm, and we define $n = |U \cup V|$.
		
	\subsection{Random subsets of balanced sets}
	A set~$X \subseteq V$ is \emph{balanced} if the graph $G[U,X]$ is not ``star-like'', with a large proportion of edges incident to a single vertex in $X$. We formally define this notion, and show that if~$X'$ is a uniformly random subset of a balanced set~$X$, then~$\eb{X}\approx 2\eb{X'}$ holds with suitably high probability.
	\begin{defn}\label{defn:balanced}
		For any real $\xi$ with $0<\xi\le 1$, a set $X \subseteq V$ is \emph{$\xi$-balanced} if every vertex in~$X$ has degree at most $\xi \eb{X}$.
	\end{defn}
	\begin{lem}\label{lem:balanced-halve}
		Let $X \subseteq V$ be a set and let $X' \subseteq X$ be a random subset formed by including each vertex of $X$ independently with probability~$\frac12$.
		\begin{enumerate}[label=(\roman*)]
			\item
			With probability at least $1-2\exp\paren*{-\abs{X}/24}$,
			we have $\abs{X'} \le \frac34 \abs{X}$.
			\item
			Let $\gamma,\xi$ be reals with $0 < \xi \le 1$ and $0<\gamma\le\frac12$.
			If $X$ is $\xi$-balanced, then with probability at least $1-2\exp\paren*{-2\gamma^2/\xi}$,
			we have
		\[
			\paren*{\tfrac{1}{2}-\gamma}\cdot\eb{X}\le \eb{X'}\le \paren*{\tfrac{1}{2}+\gamma}\cdot\eb{X}\,.
		\]
	\end{enumerate}
	\end{lem}
	\begin{proof}
		For the first claim, note that $\E(|X'|) = |X|/2$ holds, and thus by Lemma~\ref{lem:chernoff}(i) we~have
		\begin{equation*}
			\pr\left(|X'| \ge \tfrac34 \cdot |X|\right) \le \pr\left(\left||X'|-\tfrac12 \cdot |X|\right| \ge \tfrac14 \cdot |X|\right) \le 2e^{-|X|/24}\,.
		\end{equation*}

    Now we prove the second claim. For each vertex $v \in X$, let $I_v$ be the indicator random variable of the event $v \in X'$. Then~$\eb{X'}$ is a function of $\setc{I_v }{ v \in X}$, and changing a single indicator variable $I_v$ alters $\eb{X'}$ by exactly~$d(v)$. Moreover, $\mathbb{E}(\eb{X'}) = \eb{X}/2$. It therefore follows by Lemma~\ref{lem:mcdiarmid} that
		\begin{equation}\label{eqn:balanced-halve}
			\pr\left(\left|\eb{X'} - \tfrac12 \cdot\eb{X}\right| \ge
			\gamma\cdot\eb{X}\right)
			\le 2\exp\left(\frac{-2\gamma^2\eb{X}^2}{\sum_{v \in X}{d(v)}^2} \right)\,.
		\end{equation}
    Since~$X$ is $\xi$-balanced, we have
    $\sum_{v\in X} {d(v)}^2
    \le \xi\eb{X}\cdot\sum_{v\in X}d(v)=\xi \eb{X}^2$.
		With~\eqref{eqn:balanced-halve}, the claimed upper bound of $2\exp\paren*{-2\gamma^2/\xi}$ on the error probability follows.
	\end{proof}
	
    \newcommand{\sampleneighbors}{\textnormal{\texttt{SampleNeighbours}}}
	In using Lemma~\ref{lem:balanced-halve}, we will take $\gamma = \Theta(\epsilon/\log n)$ and $\xi = \Theta(\gamma^2/\log\log n)$. To motivate this choice, consider the following toy argument:
	
	Suppose simplistically that Lemma~\ref{lem:balanced-halve}(ii) was true for all sets, not just for balanced sets, and that~$\xi$ could be chosen arbitrarily. We will see later (using the {\sampleneighbors} algorithm defined in Section~\ref{sec:degrees}) that, if $\eb{X}$ is small, we can quickly determine it exactly. In this situation, the following algorithm would estimate $e(G)$: start with $X_0 = V$. Given~$X_i$, check whether $\eb{X_i}$ is small enough to determine exactly. If so, output $2^i\eb{X_i}$. If not, form $X_{i+1}$ from $X_i$ by including each element independently with probability $\tfrac{1}{2}$. Let $X_t$ be the final set formed this way.
	By Lemma~\ref{lem:balanced-halve}(i), we have $t = O(\log n)$ with high probability. By our supposed simplistic version of Lemma~\ref{lem:balanced-halve}(ii), we have $\eb{X_t} \in (1\pm \gamma)^t\eb{X_0}/2^t = (1\pm\gamma)^t e(G)/2^t$; thus the algorithm gives a valid $\epsilon$-approximation whenever $(1\pm \gamma)^t \subseteq (1\pm \epsilon)$. We have $(1\pm \gamma)^t \subseteq 1 \pm 4t\gamma$ for sufficiently small $\gamma$, so this holds for $\gamma = O(\epsilon/\log n) = O(\epsilon/t)$. Finally, using a union bound together with the fact that $t=O(\log n)$ holds with high probability, Lemma~\ref{lem:balanced-halve}(ii) holds at each stage with probability at least $1 - O(\log n)\cdot \exp(-2\gamma^2/\xi)$; this can be made arbitrarily large by taking $\xi = O(\gamma^2/\log\log n)$. 
	
	Of course, Lemma~\ref{lem:balanced-halve}(ii) is not true for all sets --- it fails badly if $G[U,X]$ is a star, for example. While the above argument does not use independence queries at all, we will need them to deal with unbalanced sets.
	
	\subsection{Estimating vertex degrees}\label{sec:degrees}
	In order to test whether a set~$X$ is balanced and thus whether taking a uniformly random subset of~$X$ will give a good approximation of~$\eb{X}$ via Lemma~\ref{lem:balanced-halve}, we will efficiently approximate the \emph{relative degrees} $d(v)/\abs{\nbh{X}}$ for all $v \in X$.
	To this end, we will use independence queries to uniformly sample a random subset~$Y\subseteq\nbh{X}$ of a given size~$y$.
	We show that, with high probability, the random variable~${\abs{N(v)\cap Y}}/{\abs{Y}}$ is a $\tfrac{1}{2}$-approximation of the relative degree unless the relative degree is smaller than~$\xi/140$, in which case ${\abs{N(v)\cap Y}}/{\abs{Y}}$ is no larger than~$\xi/20$.
	\begin{lem}\label{lem:degree-2apx}
		Let $X\subseteq V$ and let $y\in\N$ with $y\le\abs{N(X)}$.
    Let $Y\subseteq N(X)$ be a uniformly-random size-$y$ subset of~$N(X)$.
		Let $v\in X$ be a vertex and write
		\begin{equation*}
			\delta(v)=\frac{\abs{\nbh{v}}}{\abs{\nbh{X}}}
			\qquad
			\text{and}
			\qquad
			\tilde\delta(v) = \frac{\abs{\nbh{v}\cap Y}}{\abs{Y}}
			\,.
		\end{equation*}
		Let $\xi>0$. 
		If $\delta(v)\ge\xi/140$, then with probability at least $1-2\exp(-\xi y/2000)$, the number $\tilde\delta(v)$ is a $\frac12$-approximation of~$\delta(v)$.
		On the other hand, if $\delta(v)\le\xi/140$, then with probability at least $1-2\exp(-\xi y/20)$, we have $\tilde\delta(v) \le\xi/20$.
  \end{lem}
  \begin{proof}
    The random variable $|N(v) \cap Y|$ follows a hypergeometric distribution with mean $\mu_v=\delta(v)\cdot y$.
		By Lemma~\ref{lem:chernoff}(i), we have
    \begin{equation*}
      \pr\paren*{\abs[\Big]{
			\abs{N(v)\cap Y}
			-\mu_v} \ge \frac{\mu_v}{2}}
			\le
      2\exp(-\mu_v/12)\,.
    \end{equation*}
		If $\delta(v)\ge\xi/140$ and thus $\mu_v\ge \xi y/140$, this immediately implies the first claim.
		Similarly, if $\delta(v)\le \xi/140$ and thus $t:=\frac\xi{20} y\ge7\mu_v$ holds, then Lemma~\ref{lem:chernoff}(ii) immediately implies the second claim.
  \end{proof}
	
	When we use Lemma~\ref{lem:degree-2apx}, we will apply it to all $O(n)$ vertices in each of the $O(\log n)$ iterations of the overall algorithm. So in order for a union bound to give something meaningful, we need a success probability of $1-\Omega(1/(n\log n))$.
	We will therefore set $y=\Theta(\xi^{-1}\log n)=\Theta(\epsilon^{-2}\log^3 n\log\log n)$.

	We can sample a uniformly random set~$Y\subseteq\nbh{X}$, using the following straightforward procedure.
	It is the only component of our algorithm that uses independence queries.

  \begin{algorr}{\sampleneighbors}{%
    The algorithm takes as input a set~$X \subseteq V$
    and an integer~$y$, and it returns a set~$Y\subseteq U$
    such that $\abs{N(X)} < y$ implies $Y=N(X)$
    and $\abs{N(X)} \ge y$ implies that~$Y$ is a uniformly random 
    size-$y$ subset of~$N(X)$.}
    \item\label{sampleneighbors:shuffle}
      Let $u_1,\dots,u_{\abs{U}}$ be a uniformly random ordering of~$U$ and let~$Y=\emptyset$. 
    \item\label{sampleneighbors:while} While $\abs{Y}<y$:
    \begin{enumerate}[label=\texttt{\alph*}] 
    \item
      Find the smallest~$i$ with~$u_i\in N(X)\setminus Y$.
      To do so, we use independence queries of the form~$\indo{G}\paren{X\cup\set{u_1,\dots,u_j}\setminus Y}$ and perform binary search over~$j\in\set{1,\dots,\abs{U}}$.
    \item
      If~$u_i$ was found, add it to~$Y$. Otherwise we have~$Y=N(X)$ and return~$Y$.
    \end{enumerate}
    \item
    Return~$Y$.
  \end{algorr}
  
  \begin{lem}\label{lem:sample-correct}
  	The algorithm \sampleneighbors{} is correct, runs in time $O(n\log n)$, and makes at most $O(y\log n)$ independence queries.
  \end{lem}
  \begin{proof}
  	The uniform ordering of $U$ induces a uniform ordering of $N(X)$, which implies that \sampleneighbors{} is correct. For the running time, note that \aref{sampleneighbors:shuffle} runs in time~$O(n)$ (using Fisher--Yates shuffling) and each binary search runs in time~$O(\log n)$. Thus the overall running time is $O(n+y\log n) = O(n\log n)$ and the number of independence queries is~$O(y\log n)$.
  \end{proof}

	We use \sampleneighbors{} for two purposes: If it returns a set~$Y$ of size less than~$y$, then~$Y=\nbh{X}$ holds and~$Y$ is small enough to compute $\eb{X}$ using the adjacency oracle for all pairs in~$Y\times X$.
	Otherwise the set~$Y$ gives us good estimates for the relative degrees of vertices in~$X$ by Lemma~\ref{lem:degree-2apx}.
	In particular, we shall use this to approximate the set of vertices in $X$ of high relative degree, as encapsulated by the following definition.

	\begin{defn}\label{def:core}
		Let $\xi\in \R$ with $0 < \xi \le 1$ and let $X \subseteq V$.
		We say $S \subseteq X$ is a \emph{$\xi$-core} of $X$ if it satisfies the following properties:
		\begin{enumerate}[label={(W\arabic*)},leftmargin=3em]
			\item every vertex in $X$ with degree at least $\frac\xi8\cdot |\nbh{X}|$ is contained in $S$;\label{item:wit-1}
			\item every vertex in $S$ has degree at least $\frac\xi{32}\cdot|\nbh{X}|$.\label{item:wit-2}
		\end{enumerate}
	\end{defn}

	We will show in the proof of Theorem~\ref{thm:finegrain} that the estimates given by Lemma~\ref{lem:degree-2apx} do indeed yield cores.
	We now relate cores to balancedness.

	\begin{lem}\label{lem:cores-work-new}
    Let $\xi\in\R$ with $0 < \xi \le 1$ and let $S$ be a $\xi$-core of a set $X\subseteq V$.
		\begin{enumerate}[label=(\roman*)]
			\item If $\abs{S}\ge 32/\xi^2$, then $X$ is $\xi$-balanced.
			\item
			If $X \setminus S$ contains a vertex of degree at least~$\frac\xi4\cdot\abs{\nbh{X\setminus S}}$, then $|\nbh{X \setminus S}| \le \frac12 \cdot|\nbh{X}|$.
			Otherwise, $X\setminus S$ is $\frac\xi4$-balanced.
		\end{enumerate}
	\end{lem}
	\begin{proof}
		For the first claim, suppose $|S| \ge 32/\xi^2$. Then by~\ref{item:wit-2}, at least~$32/\xi^2$ vertices in $X$ have degree at least $\frac\xi{32}\cdot\abs{\nbh{X}}$. Hence $\eb{X} \ge |\nbh{X}|/\xi$ holds, and every vertex $v \in X$ satisfies $d(v) \le |\nbh{X}|\le \xi \eb{X}$. Thus $X$ is $\xi$-balanced.
		
		For the second claim, suppose $v\in X\setminus S$ is a vertex whose degree satisfies~$d(v)\ge\frac\xi4\cdot\abs{\nbh{X\setminus S}}$.
		Since $v \notin S$, we also have $d(v)\le\frac\xi{8}\cdot|\nbh{X}|$ by~\ref{item:wit-1}.
		Together, these facts imply
		$
		\abs{\nbh{X\setminus S}}
		\le \frac4\xi\cdot d(v)
		\le \frac12\cdot\abs{\nbh{X}}
		$ as required.
		Finally, note that $\abs{\nbh{X\setminus S}}\le\eb{X\setminus S}$ holds, so if all vertices in~$X\setminus S$ have degree at most $\frac\xi4\cdot\abs{\nbh{X\setminus S}}$, then $X\setminus S$ is $\frac\xi4$-balanced by definition.
	\end{proof}	
	
	\subsection{The Overall Algorithm}
	Throughout this section, we will take 
		\begin{align*}
		\gamma &= \frac{\epsilon}{800\log n}\,,\qquad
		\xi = \frac{\gamma^2}{5\log \log n} = \frac{\epsilon^2}{8\cdot 10^5\log^2 n\log \log n}\,,\text{ and}\\
		y &= \frac{4000\log n}{\xi} = \frac{32\cdot 10^8\log^3 n\log \log n}{\epsilon^2}\,.
		\end{align*}	
	The edge counting algorithm works in $O(\log n)$ iterations, starting with~$X=V$.
	In each iteration, either $|X|$ is roughly halved, or $|N(X)|$ is at least halved.
	We formulate the algorithm recursively.

	\label{algo:edgecount-new}%
	\begin{algorr}{\counter{}$(X)$}{%
		This recursive algorithm takes as input a set~$X \subseteq V$ and returns an $\epsilon$-approximation to $\eb{X}$ with suitably high probability. (Recall that the input graph $G=(U,V,E)$ and the allowed error $\epsilon>0$ have already been defined globally.)
		}
		\item\label{ec:sampleneighbors}
    Use $\sampleneighbors(X,y)$ to sample a uniformly random~$Y\subseteq N(X)$ of size $\min\{y,|N(X)|\}$.
  	\item\label{ec:trivial}
		If $\abs{X}\le 24\log n$ or $\abs{Y}<y$, then compute $\eb{X}$ using adjacency queries on $U\times X$ or $Y\times X$, respectively.
		\quad\textit{(if $\abs{Y}<y$, then $Y=\nbh{X}$ holds by the properties of \sampleneighbors{})}
		\item\label{ec:delta-apx}
		For all~$v\in X$, compute $\tilde\delta(v)=\frac{\abs{N(v)\cap Y}}{\abs{Y}}$ using adjacency queries on $Y\times X$.\\
		\textit{(w.h.p.\ each $\tilde\delta(v)$ is a $\tfrac{1}{2}$-approximation to $\delta(v)$ if $\delta(v)\ge \xi/140$)}
		\item\label{ec:core}
		Let $S=\setc{v\in X}{\tilde\delta(v)\ge\frac\xi{16}}$.\quad
		\textit{(w.h.p.\ this is a $\xi$-core)}
		\item\label{ec:if}
		If $\tilde\delta(v)\le\frac12\xi$ holds for all~$v\in X$,
		or if $\abs{S}\ge 32/\xi^2$ holds:\quad
		\textit{(w.h.p. $X$ is now $\xi$-balanced)}
		\begin{enumerate}[label=\texttt{\alph*}]
			\item\label{ec:randomdelete}
			Let~$X'$ be a uniformly random subset of $X$.
			\quad\textit{(w.h.p.\  $X'$ is at most $\frac34$ the size of~$X$)}
			\item\label{ec:rec1}
			Recursively compute $2\cdot\counter(X')$, and return this number.
		\end{enumerate}
 		\item\label{ec:ow}
		 Otherwise, independently and uniformly sample $3|U|\log n/\gamma^2$ pairs from $U \times S$, and use the adjacency oracle to determine the number $Z$ of these pairs which are edges in $G$. Let $\tilde{\partial}(S) := Z\gamma^2|S| / 3\log n$. \textit{(w.h.p.\ $\tilde{\partial}(S) \in (1\pm\gamma)\eb{S}$.)}
 		\item\label{ec:done} Return $\counter(X\setminus S)+\tilde{\partial}(S)$.
		\textit{(w.h.p.\ either~$\nbh{X\setminus S}$ is half the size of~$\nbh{X}$, or $X\setminus S$ is $\xi/4$-balanced.)}
	\end{algorr}

We are ready to formally prove our main result. 

\begin{repthm}{thm:finegrain}
	\statefinegrain{}
\end{repthm}
\begin{proof}
	We may assume without loss of generality that $n \ge 10^5$; otherwise, we simply solve the problem in $O(1)$ time by brute force using the adjacency oracle. Note that each iteration of \counter{} makes at most one recursive call, so its recursion tree is a path. An \emph{iteration} is an execution of \counter{} up to a recursive call. We first make a minor modification to \counter{}: adding a global counter to ensure that we perform at most $t=\floor{100\log n}$ iterations, otherwise halting with an output of TIMEOUT. We are very unlikely to reach this depth, but this modification will allow us to bound the running time deterministically (as required by Theorem~\ref{thm:finegrain}). Having done so, we claim that running \counter{} on input~$V$ has the claimed properties. 
	
	We first bound the running time for each iteration.
	By Lemma~\ref{lem:sample-correct}, \aref{ec:sampleneighbors} runs in time $O(n\log n)$ and makes at most $O(y \log n)$ independence queries; this step is the only one that makes independence queries at all.
	\aref{ec:trivial} takes time at most $O(n\log n)$ if $\abs{X}\le 24\log n$ or time $O(yn)$ otherwise.
	Likewise, not counting the recursive calls, \aref{ec:delta-apx}, \aref{ec:core}, and \aref{ec:if} take time $O(yn)$, and \aref{ec:ow} and \aref{ec:done} take time $O(n\log n/\gamma^2) = \epsilon^{-2}O(n\log^3 n)$. There are $O(\log n)$ total iterations, and $y=\epsilon^{-2}\Theta(\log^3n\log\log n)$, so the overall worst-case running time of the algorithm on input~$V$ is $O(yn\log n) = \epsilon^{-2}O(n\log^4 n\log\log n)$, and it makes at most $O(y\log^2 n) = \epsilon^{-2}O(\log^5 n\log\log n)$ queries to the independence oracle.
	
	\newcommand{\calF}{\mathcal{F}}
	Next, we argue that the success probability is at least $2/3$.
	To reason about this, we define the following events at each recursion depth~$1 \le i \le t$ of the algorithm:
	\begin{enumerate}
		\item[$\calF_1(i)$] Either \aref{ec:delta-apx} is not executed at depth $i$, or each $\tilde\delta(v)$ computed indeed either $\tfrac{1}{2}$-approximates $\delta(v)$ (if $\delta(v)\ge\xi/140$) or satisfies $\tilde\delta(v)\le\xi/20$ (otherwise).
		\item[$\calF_2(i)$] Either \aref{ec:if}\ref{ec:randomdelete} is not executed at depth $i$, or $\abs{X'}\le\frac34\abs{X}$ holds and the number $2\eb{X'}$ is a $2\gamma$-approximation of $\eb{X}$.
		\item[$\calF_3(i)$] Either \aref{ec:ow} is not executed at depth $i$, or $\tilde{\partial}(S)$ is a $\gamma$-approximation to $\eb{S}$.
	\end{enumerate}
	Thus $\calF_1(i)$, $\calF_2(i)$ and $\calF_3(i)$ vacuously occur if the algorithm terminates before reaching depth $i$. We write $\calF(i) = \calF_1(i) \cap \calF_2(i) \cap \calF_3(i)$, and $\calF = \bigcap_{i=1}^t \calF(i)$. We will now show that $\Pr(\calF) \ge 2/3$.
	
	Each time \aref{ec:delta-apx} is executed, the set~$Y$ returned by \sampleneighbors{} in \aref{ec:sampleneighbors} has size $y=\abs{Y}\le\abs{\nbh{X}}$, and thus this set is a uniformly random size-$y$ subset of~$\nbh{X}$.
	Lemma~\ref{lem:degree-2apx} applies and shows that each event $\calF_1(i)$ fails to occur for an individual~$v$ with probability at most $\exp(-\xi y/2000)$. By our choice of $y$, this is precisely $1/n^2$.
	Since there are at most~$n$ vertices~$v$,
	\begin{equation}\label{eqn:finegrain-F1}
		\Pr(\calF_1(i)\mbox{ fails}) \le 1/n.
	\end{equation}
	
	Conditioned on $\calF_1(i)$, we claim that the set $S$ defined in \aref{ec:core} is a $\xi$-core. If $\delta(v)\ge\xi/8$, then $\tilde{\delta}(v)$ is a valid $\tfrac{1}{2}$-approximation to $\delta(v)$, so $\tilde\delta(v)\ge\xi/16$ and thus $v$ is added to~$S$; this implies that \ref{item:wit-1} holds. Conversely, if $\delta(v)<\xi/32$, then either $\tilde{\delta}(v)$ is a $\tfrac{1}{2}$-approximation of $\delta(v)$ (in which case $\tilde\delta(v)<\xi/16$ and thus~$v$ is not added to~$S$) or $\tilde{\delta}(v) \le \xi/20$ (in which case again $v$ is not added to~$S$); this implies that \ref{item:wit-2} holds.
	
	We now claim that if \aref{ec:if}\ref{ec:randomdelete} is executed, again conditioned on $\calF_1(i)$, then $X$ is $\xi$-balanced. Suppose \aref{ec:if}\ref{ec:randomdelete} is executed; therefore either $\tilde\delta(v)\le\frac12\xi$ holds for all $v\in X$ or $\abs{S}\ge 32/\xi^2$. If $\abs{S}\ge 32/\xi^2$, then $X$ is $\xi$-balanced by Lemma~\ref{lem:cores-work-new}(i), so suppose $\tilde\delta(v)\le\frac12\xi$ for all $v\in X$. Since $\calF_1(i)$ occurs, for all $v \in X$, either $\tilde{\delta}(v)$ is a $\tfrac{1}{2}$-approximation for $\delta(v)$ or $\delta(v) < \xi/140$. In the former case, $\delta(v) \le 2\tilde{\delta}(v) \le \xi$, so $\delta(v) \le \xi$ in both cases and so $X$ is $\xi$-balanced as claimed.
	
	It follows that conditioned on $\calF_1(i)$, each time \aref{ec:if}\ref{ec:randomdelete} is executed, $\abs{X}\ge 24\log n$ and $X$ is $\xi$-balanced.
	Thus Lemma~\ref{lem:balanced-halve}(i) and (ii) apply, so $\calF_2(i)$ fails with probability at most $2\exp(-\abs{X}/24)+2\exp(-2\gamma^2/\xi)$. By our choice of $\xi$, it follows that
	\begin{equation}\label{eqn:finegrain-F2}
		\Pr(\calF_2(i)\mbox{ fails} \mid \calF_1(i)) \le \frac{2}{n} + \frac{2}{\log^{10} n}.
	\end{equation}

	Finally, conditioned on $\calF_1(i)$, each time \aref{ec:ow} is executed, $Z$ is a binomial variable with mean $\mu = 3\eb{S}\log n/\gamma^2|S|$. It follows by Lemma~\ref{lem:chernoff}(i) that for all $i$,
	\begin{align*}
		\Pr(\calF_3(i)\mbox{ fails}\mid \calF_1(i)) &= \Pr(|\tilde{\partial}(S) - \partial(S)| > \gamma\partial(S)) = \Pr(|Z-\mu| > \gamma\mu)\\
		&\le 2e^{-\gamma^2\mu/3} = 2e^{-\eb{S}\log n/|S|}.
	\end{align*}
	Since $\calF_1(i)$ occurs, $S$ is a $\xi$-core (as shown above); thus by (W2), every vertex in $S$ has positive degree, and in particular $\eb{S} \ge |S|$. Thus conditioned on $\calF_1(i)$, $\calF_3(i)$ fails with probability at most $2/n$. In conjunction with~\eqref{eqn:finegrain-F1} and~\eqref{eqn:finegrain-F2}, this implies
	\[
		\Pr(\calF(i)\mbox{ fails}) \le \frac{5}{n} + \frac{2}{\log^{10} n}.
	\]
	Since $n \ge 10^5$ and $t \le 100\log n$, this is at most $1/3t$. It follows by a union bound over all $1\le i\le t$ that $\calF$ occurs with probability at least $2/3$, as claimed.
	
	Let us now show that conditioned on $\calF$, we do not output TIMEOUT.
	We claim that in every other iteration, we multiply either $\abs{\nbh{X}}$ or $\abs{X}$ by a factor of at most~$\frac34$.
	Since $\calF_2(i)$ occurs for all $i$, it is clear that $\abs{X}$ is multiplied by a factor of at most $\frac34$ if the algorithm recurses in \aref{ec:if}\ref{ec:rec1}.
	If the algorithm recurses in \aref{ec:done}, then by Lemma~\ref{lem:cores-work-new}(ii), either we reduce $\abs{\nbh{X}}$ by at least half, or the set $X\setminus S$ is $\xi/4$-balanced. In the first case we are done, in the second case it may be that $X\setminus S$ is not significantly smaller than~$X$.
	However, as $X\setminus S$ is $\xi/4$-balanced, the condition $\tilde\delta(v)\le\xi/2$ is met for all $v\in X\setminus S$ in the very next iteration of the algorithm (where the input is $X\setminus S$), and then $X\setminus S$ is multiplied by a factor of at most $\frac34$.
	Since initially we have $\abs{X}\le n$ and $\abs{\nbh{X}}\le n$, the number of iterations is thus at most $4\log_{\frac43} n< t$ as required.
	
	It remains to prove that conditioned on $\calF$, the function call $\counter(V)$ returns an $\epsilon$-approxi\-ma\-tion for $\abs{E(G)} = \eb{V}$. Let $t'\le t$ be the total number of iterations; we will prove inductively that for all $0 \le i \le t'-1$, we have $\counter(X_{t'-i}) \in (1\pm 2\gamma)^i\eb{X_{t'-i}}$. In the last iteration, the algorithm computes $\eb{X_{t'}}$ exactly, so the claim is immediate for $i=0$. If the algorithm in iteration $t'-i$ recurses in \aref{ec:if}\ref{ec:rec1}, then since $\calF_2(t'-i)$ occurs, we have
	\begin{align*}
		\counter(X_{t'-i}) &= 2\cdot\counter(X_{t'-i+1}) \in (1\pm 2\gamma)^{i-1}\eb{X_{t'-i+1}} \subseteq (1\pm 2\gamma)^i\eb{X_{t'-i}},
	\end{align*}
	as required. If instead it recurses in \aref{ec:done}, then since $\calF_3(t'-i)$ occurs, we have
	\begin{align*}
		\counter(X_{t'-i}) &= \counter(X_{t'-i+1}) + \tilde{\partial}(S)\\
		&\in (1\pm 2\gamma)^{i-1}\eb{X_{t'-i+1}} + (1\pm \gamma)\eb{S}\\
		&\subseteq (1\pm 2\gamma)^i\big(\eb{X_{t'-i+1}} + \eb{X_{t'-i+1}\setminus X_{t'-i}}\big) = (1\pm 2\gamma)^i\eb{X_{t'-i}}.
	\end{align*}
	Thus the claim holds, and in particular 
	\[
		\counter(V) = \counter(X_1) \in (1\pm 2\gamma)^{t'-1}\eb{X_1} \subseteq (1\pm 2\gamma)^te(G).
	\]
	Since $(1-2\gamma)^{t} \ge 1-2t\gamma$ and $(1+2\gamma)^{t} \le e^{2\gamma t} \le 1+8t\gamma$, it follows that $\counter(V)$ is a $8t\gamma$-approximation of $\abs{E(G)}$.
	Since $t \le 100\log n$, by our choice of $\gamma$, this is an $\epsilon$-approximation.
\end{proof}

	\section{Applications for polynomial-time problems}\label{sec:applications}

	\subsection{3SUM}\label{sec:3sum-proofs}
	
	We formally define the problems as follows.
	
	\defproblem{\dtsum}{Three lists $A$, $B$ and $C$ of integers.}
	{Decide whether there exists a tuple $(a,b,c) \in A\times B\times C$ such that $a+b=c$.}

	\defproblem{\ctsum}{Three lists $A$, $B$ and $C$ of integers.}
	{Count the number of tuples $(a,b,c) \in A\times B\times C$ such that $a+b=c$.}

  \begin{repthm}{thm:tsum}
		\statetsum{}
	\end{repthm}
	\begin{proof}
	First we note that any bounded-error randomised algorithm for \dtsum\ must read a constant proportion of the entries in $A$, $B$ and $C$, so we can assume $T(n) = \Omega(n)$.

	Let $(A,B,C)$ be an instance of \ctsum\ and let $0<\epsilon<1$. If $\epsilon \le \tfrac{1}{n}$, then we use exhaustive search to solve the problem exactly in time $O(n^3) = O(\epsilon^{-2}T(n))$.
	In the following, we assume $\epsilon > \tfrac{1}{n}$. Let $E = \setc{(a,b) \in A \times B }{ a+b \in C}$, and let $G = (A,B,E)$. We will proceed by sorting the set~$C$ in $O(n\log n)$ time, then applying the algorithm of Theorem~\ref{thm:finegrain} to $G$ and $\epsilon$.
	
	We can evaluate $\adjo{G}(a,b)$ in time $O(\log n)$ using binary search on~$C$. Moreover, for all $X \subseteq A \cup B$, we have $\indo{G}(X) = 1$ if and only if $(X\cap A, X\cap B,C)$ is a `no' instance of \dtsum, so $\indo{G}$ can be evaluated by solving a single instance of \dtsum{}, which takes $O(n)$ time to prepare. As in the proof of Theorem~\ref{thm:growthrate}, we solve the instance by invoking the assumed randomised decision algorithm $100\log n$ times and outputting the majority answer.
	The overall algorithm is given by Theorem~\ref{thm:finegrain}. As this algorithm makes at most $\epsilon^{-2} \cdot O(\log^6 n)\le O(n^2\log^6 n)$ queries to $\indo{G}$, the probability that at least one of them is answered incorrectly by the boosted randomised procedure remains negligible, at most~$O(1/n)$ by Lemma~\ref{lem:chernoff}(i), which is in particular at most $1/3$ as required.
	The overall running time is:
	\begin{align*}
		\underbracket{O(n\log n)}_{\text{sort $C$}}
		+
		\underbracket{%
		\epsilon^{-2} O(n\log^4n\log\log n)
		}_{\text{\# queries to $\adjo{G}$}}
		\cdot
		\underbracket{%
		O(\log n)
		}_{\text{binary search}}
		+
		\underbracket{%
		\epsilon^{-2} O(\log^5n\log\log n)
		}_{\text{\# queries to $\indo{G}$}}
		\cdot
		\underbracket{%
		(
		O(n)
		+
		T(n)\log n
		)
		}_{\text{prepare and solve query}}
		\,.
	\end{align*}
	We have constructed an $\epsilon$-approximation algorithm for $\dtsum{}$ that has the claimed running time.
	\end{proof}
	
	\begin{repthm}{thm:tsum-algo}
		\statetsumalgo{}
	\end{repthm}
	\begin{proof}
    Say a set $S \subseteq \Z$ is \emph{$(n,\delta)$-clustered} if it can be covered by at most $n^{1-\delta}$ intervals of length~$n$; note that it can be checked in quasilinear time whether a set is $(n,\delta)$-clustered. Let $(A,B,C)$ be an instance of \ctsum\ in which at least one of $A$, $B$ or $C$ is $(n,\delta)$-clustered. By negating and permuting sets if necessary, we may assume that $C$ is $(n,\delta)$-clustered. Exactly as in the proof of Theorem~\ref{thm:tsum}, any randomised $T(n)$-time algorithm for \dtsum\ on such instances yields a $T(n)\cdot \epsilon^{-2}O(\log^6 n\log \log n)$-time randomised approximation scheme. (In particular, note that $(X \cap A, X\cap B, C)$ remains an instance of the restricted problem.) Chan and Lewenstein~\cite[Corollary~4.3]{CL-3SUM} provide a randomised $O(n^{2-\delta/7})$-time algorithm for \dtsum\ on such instances, so the result follows.
	\end{proof}
	
	\subsection{Orthogonal Vectors}\label{sec:ov-proofs}
	
	We formally define the problems as follows.
	
	\defproblem{\dov}{Two lists $A$ and $B$ of zero-one vectors in $\R^d$.}
	{Decide whether there exists a pair $(\vec{u},\vec{v}) \in A \times B$ such that $\sum_{i=1}^d \vec{u}_i\vec{v}_i = 0$.}

	\defproblem{\cov}{Two lists $A$ and $B$ of zero-one vectors in $\R^d$.}
	{Count the number of pairs $(\vec{u},\vec{v}) \in A \times B$ such that $\sum_{i=1}^d \vec{u}_i\vec{v}_i = 0$.}

  \begin{repthm}{thm:ov}
    \stateov{}
  \end{repthm}
	\begin{proof}
	Let $(A,B)$ be an instance of \cov\ and let $0<\epsilon<1$. If $\epsilon \le n^{-2}$ then we can solve the problem exactly in time $O(n^2) = O(\epsilon^{-1})$, so suppose $\epsilon > n^{-2}$. Let $E = \setc{(a,b) \in A \times B }{ \langle a,\,b \rangle = 0}$, and let $G = (A,B,E)$ be a bipartite graph. We will proceed by applying the algorithm of Theorem~\ref{thm:finegrain} to $G$ and $\epsilon$.
	
	We can evaluate $\adjo{G}$ in $O(d)$ time by calculating the inner product. Moreover, for all $X \subseteq A \cup B$, $\indo{G}(X) = 1$ if and only if $(A \cap X,B\cap X)$ is a `no' instance of \dov, so $\indo{G}$ can be evaluated by solving a single instance of \dov\ which takes $O(nd)$ time to prepare. As in the proof of Theorem~\ref{thm:tsum}, we do so by invoking our randomised decision algorithm $100\log n$ times and outputting the majority answer. Our overall running time is then
	\[
		\epsilon^{-2}\cdot O(n\log^4n\log\log n)\cdot O(d) + \epsilon^{-2}\cdot (O(nd)+T(n,d)\log n)\cdot O(\log^5n\log\log n).
	\]
	Since any randomised algorithm for \dov\ must examine a constant proportion of the coordinates of vectors in $A$ and $B$, we have $T(n,d) = \Omega(nd)$, so the result follows.
	\end{proof}
	
	In the following definitions, $\mathcal{R}$ is a constant finite ring.
	
	\defproblem{\dov($\mathcal{R}$)}{Two lists $A$ and $B$ of vectors in $\mathcal{R}^d$.}
	{Decide whether there exists a pair $(\vec{u},\vec{v}) \in A \times B$ such that $\sum_{i=1}^d \vec{u}_i\vec{v}_i = 0_{\mathcal{R}}$.}

	\defproblem{\cov($\mathcal{R}$)}{Two lists $A$ and $B$ of vectors in $\mathcal{R}^d$.}
	{Count the number of pairs $(\vec{u},\vec{v}) \in A \times B$ such that $\sum_{i=1}^d \vec{u}_i\vec{v}_i = 0_{\mathcal{R}}$.}

  \begin{repthm}{thm:ov-algo}
    \stateovalgo{}
  \end{repthm}
  \begin{proof}
		Exactly as in the proof of Theorem~\ref{thm:ov}, any randomised $T(n,d)$-time algorithm for \dov($\mathcal{R}$) yields a $T(n,d)\cdot \epsilon^{-2}O(\log^6 n\log \log n)$-time randomised approxi\-mation scheme for \cov($\mathcal{R}$). (Note that $\mathcal{R}$ is finite and part of the problem specification, so arithmetic operations require only $O(1)$ time.) The result therefore follows from Theorems 1.6 and 1.3 (respectively) of Williams and Yu~\cite{WY-OV-algo}.
	\end{proof}
		
	\subsection{Negative-Weight Triangles}\label{sec:nwt-proofs}
	
	We formally define the problems as follows.
	
	\defproblem{\dnwt}{A tripartite graph $G$ and a symmetric function $w:V(G)^2\rightarrow\mathbb{Z}$.}
	{Decide whether there exists a triangle $abc$ in $G$ such that $w(a,b) + w(b,c) + w(c,a) < 0$.}

	\defproblem{\cnwt}{A tripartite graph $G$ and a symmetric function $w:V(G)^2\rightarrow\mathbb{Z}$.}
	{Count the number of triangles $abc$ in $G$ such that $w(a,b) + w(b,c) + w(c,a) < 0$.}

  \begin{repthm}{thm:nwt}
    \statenwt{}
  \end{repthm}
	\begin{proof}
	Let $(G,w)$ be an instance of \cnwt, let $A$, $B$ and $C$ be the vertex classes of $G$, and let $0<\epsilon<1$. If $\epsilon \le n^{-3}$ then we can solve the problem exactly in time $O(n^3) = O(\epsilon^{-1})$, so suppose $\epsilon > n^{-3}$. Let $U = A$, let $V = \setc{e \in E(G) }{ e \subseteq B \cup C}$, and~let 
	\[
		E = \setc[\Big]{(a,\{b,c\}) \in U \times V }{ \{a,b\},\{a,c\} \in E(G)\mbox{ and }w(a,b)+w(b,c)+w(c,a) < 0}.
	\]
	Let $H = (U,V,E)$, so that $H$ is a bipartite graph. We will proceed by applying the algorithm of Theorem~\ref{thm:finegrain} to $H$ and $\epsilon$. 
	
	We can evaluate $\adjo{H}$ in $O(1)$ time by summing the appropriate weights. Moreover, for all $X \subseteq U \cup V$, define a graph $G_X$ by $V(G_X) = (X \cap A) \cup B \cup C$ and 
	\[
		E(G_X) = \setc[\Big]{e \in E(G) }{ e \cap X \cap A \ne \emptyset \mbox{ or } e \in X \cap V}\,.
	\]
	Let $w_X = w|_{V(G_X)^2}$. Then for all $X \subseteq U \cup V$, $\indo{H}(X)=1$ if and only if $(G_X, w_X)$ is a `no' instance of \dnwt, so $\indo{G}$ can be evaluated by solving a single instance of \dnwt\ which takes $O(n^2)$ time to prepare. As in the proof of Theorem~\ref{thm:tsum}, we do so by invoking our randomised decision algorithm $100\log n$ times and outputting the majority answer. Our overall running time is then
	\[
		\epsilon^{-2}\cdot O(n^2\log^4n\log\log n)\cdot O(1) + \epsilon^{-2}\cdot (O(n^2)+T(n)\log n)\cdot O(\log^5n\log\log n).
	\]
	If $G$ is a complete tripartite graph, then any randomised algorithm for \dnwt\ must examine a constant proportion of the edges of $G$, so we have $T(n) = \Omega(n^2)$ and the result~follows.	
	\end{proof}
	
	In order to approximate algorithm for \#NWT, we will reduce to APSP and apply the algorithm of Williams~\cite{Williams-APSP}. We formally define APSP as follows.
	
	\defproblem{APSP}{A directed graph $G$ and a function $w:E(G)\rightarrow\Z$ such that $G$ contains no negative-weight cycles under $w$.}{Output the matrix $A$ such that for all $u,v\in V(G)$, $A_{u,v}$ is the minimum weight of any path from $u$ to $v$ in $G$.}
  \begin{repthm}{thm:nwt-algo}
    \statenwtalgo{}
	\end{repthm}
	\begin{proof}
		By Williams~\cite[Theorem 1.1]{Williams-APSP}, an $n$-vertex instance of APSP with polynomially bounded edge weights can be solved in time $n^3/e^{\Omega(\sqrt{\log n})}$. There is a well-known reduction from \dnwt\ to APSP with only constant overhead, which we give explicitly in the following paragraph. Theorem~\ref{thm:nwt} then implies the existence of an $\epsilon$-approximation algorithm for \#NWT with running time $\epsilon^{-2}n^3/e^{\Omega(\sqrt{\log n})}$, noting that the polylogarithmic overhead is subsumed into the $e^{\Omega(\sqrt{\log n})}$~term.
		
		It remains only to reduce NWT to APSP. Let $(G,w)$ be an instance of \dnwt, writing $G=(V,E)$. Form an instance $(G',w')$ of APSP as follows. Let $V(G') = (V \times [3])$, and let 
		\[
			E(G') = \bigcup_{i\in\{1,2\}} \bigcup_{\{u,v\} \in E} \{((u,i),(v,i+1)), ((v,i),(u,i+1))\}.
		\]
		Let $w'(\{(u,i),(v,i+1)\}) = w(u,v)$ for all $\{(u,i),(v,i+1)\} \in E(G')$. Thus for all $\{u,v\} \in E$, each path $(u,1)(w,2)(v,3)$ from $(u,1)$ to $(v,3)$ in $G'$ corresponds exactly to the triangle $uvw$ in $G$, and $uvw$'s weight is the length of the corresponding path plus $w(u,v)$. Let~$A$ be the output of APSP on $G'$. Then from the discussion above, $(G,w)$ is a `yes' instance of \dnwt\ if and only if for some $\{u,v\} \in E(G)$, we have $A_{(u,1),(v,3)} + w(u,v) < 0$.
		This can be checked in $O(n^2)$ time.
	\end{proof}

\begin{acks}
  We thank Rahul Santhanam and Ryan Williams for some valuable discussions.
    
  Part of this work was done while the authors were visiting the Simons Institute for the Theory of Computing. The research leading to these results has received funding from the \grantsponsor{}{European Research Council (ERC)}{} under the European Union's Seventh Framework Programme (FP7/2007--2013) ERC grant agreement no.~\grantnum{}{334828}. The paper reflects only the authors' views and not the views of the ERC or the European Commission. The European Union is not liable for any use that may be made of the information contained therein.
\end{acks}

  \bibliographystyle{ACM-Reference-Format}
	\bibliography{approximation-to-decision}
\end{document}